\documentclass[11pt]{article}
\usepackage{geometry}
\usepackage{authblk}
\usepackage{latexsym,amsmath,amsfonts,amssymb,amsbsy,amsopn,amstext,amsthm}
\usepackage{epsfig,epstopdf,graphicx}
\usepackage[mathscr]{eucal}
\usepackage{multicol,makeidx}
\usepackage[monochrome]{xcolor}
\usepackage{bm,mathrsfs}
\usepackage{float}
\usepackage{stackengine}
\usepackage{tablefootnote}
\usepackage{threeparttable}
\usepackage{mathdots}
\usepackage{tabu}
\usepackage{enumerate}
\usepackage[labelformat=simple]{subcaption}

\def\eq{\displaystyle\stackrel\triangle=}
\def\xra{\xrightarrow}

\DeclareMathOperator*{\argmin}{arg\,min}

\newcommand*\mcapinn[2]{\vcenter{\hbox{$\mathsurround=0pt
  \ifx\displaystyle#1\textstyle\else#1\fi\bigcap$}}}

\newcommand*\mcupinn[2]{\vcenter{\hbox{$\mathsurround=0pt
  \ifx\displaystyle#1\textstyle\else#1\fi\bigcup$}}}

\DeclareFontFamily{OT1}{pzc}{}
\DeclareFontShape{OT1}{pzc}{m}{it}{<-> s * [1.200] pzcmi7t}{}
\DeclareMathAlphabet{\mathpzc}{OT1}{pzc}{m}{it}

\graphicspath{{./images/}}

\newtheorem{assum}{Assumption}
\newtheorem{thm}{Theorem}
\newtheorem{prop}{Proposition}
\newtheorem{lem}{Lemma}
\newtheorem{rem}{Remark}

\providecommand{\stackanchor}[2]{\shortstack{#1\\#2}}
\newenvironment{keyword}
{\par\small\noindent\textbf{\textit{Keywords---}}}
{\par}

\begin{document}

	\title{\textbf{Sequentially decoupling estimators for Box-Jenkins model  estimation}}
	\author{Biqiang Mu}

	\affil{Academy of Mathematics and Systems Science\\
		 Chinese Academy of Sciences, Beijing 100190, China\\
	\texttt{bqmu@amss.ac.cn}}

\date{}
\maketitle

	\begin{abstract}                          
		In this paper, we propose {\color{blue}a consistent and asymptotically efficient estimation method for Box–Jenkins (BJ) models that is applicable under both open-loop and closed-loop data conditions, serving as a possible alternative to the weighted null-space fitting  approach.
				The method comprises two stages: an initial sequentially decoupling (SD) estimator, followed by  Gauss–Newton (GN) refinement step.}
		The SD estimator  {\color{blue}is constructed from three sequential least squares (LS) estimators:}
		(i) estimation of a high-order autoregressive model with exogenous inputs (ARX) model;
		(ii) estimation of the BJ model’s dynamic model via an auxiliary output-error (OE) model; and
		(iii) estimation of the noise model of the BJ model using another auxiliary  OE model.
		We establish the consistency of the SD estimator under standard regularity conditions, leveraging the consistency of the underlying LS estimators for both the ARX and OE models.  {\color{blue}Moreover, we show that one-step GN iteration starting from the SD estimator yields an estimator that} is asymptotically equivalent to the prediction error method, provided the ARX model order satisfies a mild growth condition.
		Simulation studies confirm the theoretical properties of the proposed method.

	\end{abstract}


\section{Introduction}

System identification seeks to establish accurate mathematical models for practical dynamic systems using measured data. Among various model structures, the family of linear models, which describes the linear relationship between input and output, plays a crucial role in system identification due to its simplicity and  importance \cite{Ljung1999,Soderstrom1989}.
In particular, the  BJ model consists of two parametric models characterized by rational functions: the dynamic model and the noise model, and has attracted significant attention from both theoretical and practical perspectives. This model’s ability to capture complex system dynamics and noise structures makes it highly valuable \cite{Ljung1999,Box2015, Pintelon2006J2, Triolo1988}. The BJ model includes several widely used special cases, such as the finite impulse response  model, the ARX model, the  OE  model, and the autoregressive moving average with exogenous inputs  model. Furthermore, more complex models based on the BJ structure have been explored, including linear parameter-varying BJ models \cite{Laurain2010} and jump BJ models \cite{Piga2020}.


Several methods have been developed to estimate the unknown parameters of the four polynomials that define the BJ model. Among them, the prediction error method (PEM) is the most widely used in practice. The conventional time-domain PEM minimizes the sum of squared prediction errors and is particularly attractive due to its asymptotic efficiency under Gaussian noise and its readily available implementation in MATLAB’s System Identification Toolbox \cite{Ljung1999, Ljung2012}.
PEM has also been extended to the frequency domain via maximum likelihood (ML) formulations, enabling BJ model estimation in both open- and closed-loop settings \cite{Ljung1993, Mckelvey2002,Pintelon2006J1}.
However, a major challenge  of the PEM  is the non-convexity of the underlying optimization problem. This often makes it difficult to guarantee convergence to the global optimum, as iterative solvers may become trapped in local minima.

To address {\color{blue}the non-convex optimization problem inherent in the PEM}, the refined instrumental variable (RIV) method introduced in \cite{Young2015} employs an iterative pseudo-linear regression algorithm derived from the first-order optimality conditions of the  ML criterion.
	{\color{blue}Another prominent approach to BJ model estimation falls within the class of model reduction methods \cite{Wahlberg1989}. These are multi-step procedures that share a common first step: estimation of a high-order ARX model via  least squares (LS) estimation. The theoretical foundation of this class of methods rests on the fact that the high-order ARX estimate and its covariance constitute a sufficient statistic for the parameters of the underlying BJ model \cite{Lehmann1998}. The methods differ in their subsequent model reduction steps, which aim to extract the BJ parameters from this nonparametric ARX estimate.}

The Box–Jenkins Steiglitz–McBride (BJSM) method \cite{Zhu2016} {\color{blue}is designed for open-loop data and proceeds in two stages:}
(i) it estimates the dynamic model by applying Steiglitz–McBride (SM) iterations to an output-error (OE) model {\color{blue}constructed from filtered input–output signals derived from  the nonparametric  ARX estimate; and}
(ii) it estimates the noise model by fitting an ARMA model {\color{blue}via a nonconvex optimization procedure.
		Building on this framework, \cite{Everitt2018} proposed the model order reduction Steiglitz–McBride (MORSM) method, which improves upon BJSM in terms of convergence properties. Under open-loop conditions, MORSM achieves consistency and asymptotic efficiency for the dynamic model of Box–Jenkins systems using just one-step SM iteration.
		More recently, the weighted null-space fitting (WNSF) method  \cite{Galrinho2019J1}  was introduced to provide a consistent and asymptotically efficient estimate of both the dynamic and noise components of BJ models, applicable to both open- and closed-loop data.
		WNSF avoids nonconvex optimization and iterative schemes by relying on a multi-step (weighted) LS procedure.
		Owing to its strong theoretical guarantees and computational efficiency, WNSF has been successfully extended to a wide range of identification problems, including:
		BJ models with nonparametric noise models \cite{Galrinho2019J2},
		multi-input multi-output (MIMO) BJ models \cite{Galrinho2018},
		recursive identification of MIMO BJ systems \cite{Fang2021},
		dynamic network identification \cite{Galrinho2018, Fonken2022, Kivits2023}, and
		subspace identification \cite{He2024}.

	Although the WNSF method enjoys asymptotic efficiency guarantees, we found from simulations that its estimation accuracy degrades significantly under low-pass input excitation. We conjecture that this degradation may arise from the weight matrix in the final weighted least squares step, which becomes severely ill-conditioned under low-pass input excitation. 
		To address this limitation, we propose an alternative method, denoted by SDGN, for  BJ model estimation that is applicable to both open-loop and closed-loop data. The SDGN method consists of two stages:
		(i) an SD estimator, which provides a consistent (though not necessarily efficient) initial estimate; and
		(ii) a GN refinement, which elevates this initial estimate to asymptotic efficiency.
		A key advantage of this two-step strategy is that it simplifies the design of the initial estimator: consistency alone is sufficient, and  its precise convergence rate is not important.
		The SD estimator itself belongs to the class of model reduction methods. It leverages the nonparametric  ARX estimate to construct filtered input–output signals, and then sequentially recovers the dynamic and noise components of the BJ model  by solving two auxiliary OE models via LS estimation.

		The proposed SDGN method is proved to be consistent and asymptotically efficient under Gaussian noise in both open- and closed-loop settings, matching the theoretical guarantees of WNSF.
		Unlike BJSM and MORSM, which are designed for  open-loop data and lack   theoretical characterization for the noise model, SDGN provides a complete and unified treatment.
		Crucially, SDGN can avoids the potential ill-conditioning issues that affect WNSF under low-pass inputs, as the model reduction steps of the SD do not involve  the same possibly ill-conditioned matrix.
		Our simulation results  verify that SDGN consistently outperforms WNSF in low-excitation scenarios, while achieving comparable  performance in other settings.
	}



The rest of the paper is organized as follows. Section \ref{sec2} introduces the BJ model, and {\color{blue}the general framework of asymptotically efficient} two-step estimators. Section \ref{sec3} presents the estimation procedures for the ARX($\infty$) and OE models, which {\color{blue}form the building blocks of} the SD estimator. Section \ref{sec4} develops the SD estimator {\color{blue}for  both open-loop and closed-loop settings,} establishes its consistency, {\color{blue}and proves the asymptotic efficiency of the proposed SDGN method.} Section \ref{sec5} evaluates the performance of the SD and {\color{blue}SDGN estimators through comprehensive Monte Carlo simulations, comparing them against state-of-the-art methods.} Finally, Section \ref{sec6} concludes the paper with a brief {\color{blue}summary and outlook.}

	{\color{blue}\textbf{Notation}
		We use the following notation throughout the paper.
		The symbol $q$ denotes the forward operator acting on a sequence $\{u(t),t\geq 1\}$, i.e., $qu(t)=u(t+1)$.
		The symbol $E$ means the mathematical expectation of a random variable.
		For a random sequence $\{X(n),n\geq 1\}$, (i) $X(n)=O_p(1)$ represents that $\{X(n)\}$ is bounded in probability, i.e.,  for any $\epsilon> 0$, there exists constant $L>0$ and integer $N>0$ such that $P(|X(n)|>L)<\epsilon$ for $n>N$; (ii) $X(n)\xra{}0$ (equivalently $X(n)=o_p(1)$) represents that $\{X(n)\}$ converges to zero in probability, i.e.,
		for any $\epsilon> 0$,  there holds that $\lim_{n\xra{}\infty}P(|X(n)|>\epsilon)=0$;
		(iii) $X(n)\xra{}\mathcal{N}(0,\sigma^2)$ denotes that $X(n)$ converges in distribution to a Gaussian random variable with mean zero and variance $\sigma^2$;
		(iv) $\sigma\{X(n),1\leq n\leq t\}$  means the $\sigma$-algebra generated by  random variables $\{X(n),1\leq n\leq t\}$.
		For a vector $a$, $\|a\|$ and $\|a\|_1$ means the 2-norm and 1-norm, respectively.
		For a square matrix $A$,    $A > 0$  means $A$ is positive definite.
	}


\section{Problem formulation}
\label{sec2}

Consider the single-input single-output BJ model \cite[Equation (4.31)]{Ljung1999},  described by
\begin{subequations}
	\label{gm}
	\begin{align}
		y(t)=\frac{B(q)}{F(q)}u(t)+\frac{C(q)}{D(q)}e(t)
	\end{align}
	\begingroup\allowdisplaybreaks
	with
	\begin{align}
		 & B(q)=b_1q^{-1}+\cdots +b_{p_b}q^{-{p_b}},   \\
		 & C(q)=1+c_1q^{-1}+\cdots +c_{p_c}q^{-{p_c}}, \\
		 & D(q)=1+d_1q^{-1}+\cdots +d_{p_d}q^{-{p_d}}, \\
		 & F(q)=1+f_1q^{-1}+\cdots +f_{p_f}q^{-{p_f}},
	\end{align}
\end{subequations}
\endgroup
where  $y(t)$, $u(t)$ and $e(t)$ are the output, input and noise at time $t$, and the rational functions $B(q)/F(q)$ and $C(q)/D(q)$ are called the resulting  dynamic model and noise model, respectively.

Let us collect all the parameters of the  model \eqref{gm} in
$\theta=[\theta_b^T,\theta_c^T,\theta_d^T,\theta_f^T]^T$ with  $\theta_b=[b_1,\cdots,b_{p_b}]^T$, $\theta_c=[c_1,\cdots,c_{p_c}]^T$, $\theta_d=[d_1,\cdots,d_{p_d}]^T$,  and $\theta_f=[f_1,\cdots,f_{p_f}]^T$.
Let $\theta^o$ be the true parameters of   the model \eqref{gm}.
Accordingly, $\theta_b^o$, $\theta_c^o$, $\theta_d^o$ , and $\theta_f^o$ are the true parameters  corresponding to polynomials $B^o(q), C^o(q),$ $D^o(q)$, and $F^o(q)$.
Thus, the estimation of the BJ model \eqref{gm} aims to recover the parameters $\theta^o$ as accurately as possible based on the available data $\{u(t),y(t),t=1,\cdots,n\}$.

\subsection{Assumptions}
Let us first list the assumptions on the BJ  model \eqref{gm} as follows.
\begin{assum}
	(True model)
	\label{assum1}
	\begin{enumerate}[(i)]
		\item The orders $p_b,p_c,p_d,$ and $p_f$  are available.

		\item All of the true polynomials $C^o(q)$, $D^o(q)$, and $F^o(q)$ are stable, i.e., all the roots of these three polynomials are inside of the unit circle.

		\item The polynomials $q^{p_b}B^o(q)$ and $q^{p_f}F^o(q)$ have no common factor.

		\item The polynomials $q^{p_c}C^o(q)$ and $q^{p_d}D^o(q)$ have no common factor.

	\end{enumerate}
\end{assum}

\begin{assum} (Noise)
	\label{assum3}
	The noise sequence $\{e(t)\}$ is a stochastic process that satisfies
	\begin{align*}
		E(e(t)|\mathscr{F}_{t-1})=0, 	E(e(t)^2|\mathscr{F}_{t-1})=\sigma^2, E(|e(t)|^{10})\leq C,
	\end{align*}
	where $C$ is a constant and  $\mathscr{F}_t$ is the $\sigma$-algebra generated {\color{red!70!black}according to}  $\{e(s),u(s),1\leq s\leq t\}$.

\end{assum}

{\color{red!70!black}
\begin{assum} (Input)
	\label{assum4}
	The input $\{u(t)\}$ has a feedback form  	$u(t)=-K(q)y(t)+r(t)$ satisfying the following conditions:
	\begin{enumerate}[(i)]
		\item   The sequence \(\{r(t)\}\) is independent of \(\{e(t)\}\), \( \sqrt{\log n/n}\)-quasi-stationary, and uniformly bounded.

		\item
		      Let \(\Phi_r(q) = F_r(q)F_r(q^{-1})\) be the spectral factorization of  \(\{r(t)\}\) with causal  \(F_r(q)\).
		      Then  \(F_r(q)\)  is  BIBO stable.

		\item The closed loop system is \(1/\sqrt{n}\)-stable.

		\item The feedback transfer function \(K(q)\) is bounded on the unit circle.

		\item  The  spectral density of the process \([r_t~ e_t ]^T\) is  bounded from below by the matrix \(\psi I\) with  \(\psi > 0\).
	\end{enumerate}
\end{assum}
The detailed definitions on quasi-stationarity, stable, and  spectral density of a sequence can refer to   \cite{Ljung1992}.

\begin{rem}
	Assumptions \ref{assum1}--\ref{assum4}  above on true model, noise and input are standard for BJ model estimation operated in open- and closed-loop (See\cite[Assumptions 1--3]{Galrinho2019J1}).
\end{rem}
}


\subsection{Prediction error methods}
By  \cite[Equation (4.32)]{Ljung1999}, the one-step-ahead predictor for the model \eqref{gm} is
\begin{align}
	\widehat{y}(t|\theta)
	=\frac{D(q)B(q)}{C(q)F(q)}u(t)
	+\left(
	1-
	\frac{D(q)}{C(q)}
	\right)y(t)
\end{align}
and hence the prediction error of  the model \eqref{gm} is
\begin{align}
	\varepsilon(t,\theta)=y(t) \!-\! \widehat{y}(t|\theta)=
	\frac{D(q)}{C(q)}y(t)\!-\!\frac{D(q)B(q)}{C(q)F(q)}u(t),
\end{align}
where $y(t)$ is the output of the model \eqref{gm} with the true parameters $\theta^o$ under $u(t)$ and $e(t)$.
Thus, we obtain  the  loss function for the PEM with a quadratic form
\begin{align}\ell_n(\theta) = \frac{1}{n}\sum_{t=1}^n \varepsilon^2(t,\theta)\label{pemloss}
\end{align}
and the PEM estimates
the   true parameters $\theta^o$ by minimizing the loss function
\begin{align}
	\label{pem}
	\widehat{\theta}_n^{\rm pem}\eq \argmin_{\theta}
	\ell_n(\theta).
\end{align}
Under certain conditions, the PEM  $\widehat{\theta}_n^{\rm pem}$ enjoys attractive convergence properties, illustrated below. Before presenting it, we  need an assumption on the differentiability of $\ell_n(\theta)$.
{\color{red!70!black}
\begin{assum}
	\label{assum5}
	The loss function $\ell_n(\theta)$ is three-times differentiable on a compact set of $\theta^o$ and its  Hessian matrix $\frac{\partial^2 \ell_n(\theta)}{\partial \theta\partial \theta^T}\Big|_{\theta = \theta^o}$
	exists and converges to a positive definite matrix in probability as $n\xra{}\infty$. Moreover, $\big|\frac{\partial^3 \ell_n(\theta)}{\partial \theta_i\partial \theta_j\partial \theta_k} \big|$ are bounded in probability on a compact set of $\theta^o$.
\end{assum}
}
\begin{prop}
	\label{prop11}
	\cite[Theorem 8.2, page 254, Theorem 9.1, page 282]{Ljung1999}
	Suppose that Assumptions \ref{assum1}--\ref{assum4} hold.
	Thus, the PEM  estimator $\widehat{\theta}_n^{\rm pem}$ converges to its true value $\theta^o$  in probability as $n\xra{}\infty$.
	If further Assumption \ref{assum5} holds,  then
	$\widehat{\theta}_n^{\rm pem}$ shares the asymptotic normality
	\begin{align}
		\sqrt{n}(\widehat{\theta}_n^{\rm pem}-\theta^o) \xra{}\mathscr{N}(0,\sigma^2(M^o)^{-1}),
	\end{align}
	where
	$M^o\eq E\big(	\psi(t,\theta^o)\psi(t,\theta^o)^T\big)$
	with $	\psi(t,\theta^o)=-\frac{\partial \varepsilon(t,\theta)}{\partial \theta}|_{\theta=\theta^o}$.
\end{prop}

\begin{rem}
	The PEM estimator \eqref{pem} with the quadratic loss function is equivalent to the ML estimation and thus becomes asymptotically efficient when the noise is a zero-mean independent and identically distributed (iid) Gaussian random variable sequence {\color{blue}\cite[Section 9.7, page 304]{Ljung1999}.}
\end{rem}

\subsection{{\color{blue}Asymptotically efficient} two-step estimators}
{\color{blue}
	The optimization problem   \eqref{pem} is non-convex, making it difficult to locate the global minimum without a good initial value. However, if a high-quality initial estimator can be constructed from the data, convergence to the global minimum becomes attainable.
	This principle underlies the class of asymptotically efficient two-step estimators \cite{Vaart1998, Mu2017J1, Lehmann1998}, which proceed as follows:
}
\begin{enumerate}[Step 1]
	\item  Construct a consistent  initial    estimator $\widehat{\theta}_n$ for $\theta^o$ using the observed data $\{u(t),y(t),t=1,\cdots,n\}$, i.e.,
	      $\widehat{\theta}_n-\theta^o=o_p(1)$;
	\item  Run a Newton-based optimization algorithm for   problem \eqref{pem} with $\widehat{\theta}_n$  as the starting point.
\end{enumerate}

In Step 2,  Newton-based optimization algorithms, such as the  GN algorithm, Newton--Raphson algorithm, or Levenberg--Marquardt algorithm,  can refine the consistent initial estimator {\color{blue} to achieve asymptotic efficiency. Among these,} the GN algorithm is often preferred {\color{blue}due to its favorable balance of computational simplicity} and strong theoretical properties, as detailed in Lemma~\ref{lm1} below {\color{blue}\cite{Jennrich1986, Brockwell1991, Duchesne2020}.}
Let the gradient of the one-step-ahead predictor $\widehat{y}(t|\theta)$ with respect to $\theta$ be denoted by
\begin{align}
	\frac{\partial \widehat{y}(t|\theta)}{\partial \theta}
	=\left[\frac{\partial \widehat{y}(t|\theta)}{\partial \theta_b^T},
		\frac{\partial \widehat{y}(t|\theta)}{\partial \theta_c^T},
		\frac{\partial \widehat{y}(t|\theta)}{\partial \theta_d^T},
		\frac{\partial \widehat{y}(t|\theta)}{\partial \theta_f^T}
		\right]^T
\end{align}
with
\begingroup\allowdisplaybreaks
\begin{align*}
	 & \frac{\partial \widehat{y}(t|\theta)}{\partial \theta_b^T}
	= \frac{D(q)}{C(q)F(q)}[q^{-1},q^{-2},\cdots,q^{-p_{b}}]u(t),                \\
	 & \frac{\partial \widehat{y}(t|\theta)}{\partial \theta_c^T}
	= -\frac{B(q)D(q)}{C^2(q)F(q)}[q^{-1},q^{-2},\cdots,q^{-p_{c}}]u(t)        +\frac{D(q)}{C^2(q)}[q^{-1},q^{-2},\cdots,q^{-p_{c}}]y(t), \\
	 & \frac{\partial \widehat{y}(t|\theta)}{\partial \theta_d^T}
	= \frac{B(q)}{C(q)F(q)}[q^{-1},q^{-2},\cdots,q^{-p_{d}}]u(t) -\frac{1}{C(q)}[q^{-1},q^{-2},\cdots,q^{-p_{d}}]y(t) ,      \\
	 & \frac{\partial \widehat{y}(t|\theta)}{\partial \theta_f^T}
	= -\frac{B(q)D(q)}{C(q)F^2(q)}[q^{-1},q^{-2},\cdots,q^{-p_{f}}]u(t)	.
\end{align*}
\endgroup
{\color{blue}Thus, the
	one-step GN refinement  is}
\begin{subequations}\label{gn}
	\begin{align}
		           & \widehat{\theta}_n^{gn}=\widehat{\theta}_n+(J^T J )^{-1}J^T (y-f ), \\
		\!\!\!\!J  & =\left[
			\frac{\partial \widehat{y}(1|\theta)}{\partial \theta},
			\cdots,
			\frac{\partial \widehat{y}(n|\theta)}{\partial \theta}
		\right]^T\Big|_{\theta=\widehat{\theta}_n},                                      \\
		\!\!	\!\!y & = [y(1),\cdots,y(n)]^T,
		f =[\widehat{y}(1| \widehat{\theta}_n ),\cdots,\widehat{y}(n| \widehat{\theta}_n )]^T,
	\end{align}
\end{subequations}
where   $\widehat{\theta}_n$ is the consistent estimate given by Step 1.
The  two-step estimator presented above has the following attractive properties.
	{\color{red!70!black}
		\begin{lem}\cite[Theorem 2]{Duchesne2020}
			\label{lm1}
			Suppose that Assumptions \ref{assum1}--\ref{assum5} hold.
			If  the initial estimator $\widehat{\theta}_n$ is consistent with the rate of convergence $\widehat{\theta}_n-\theta^o=O_p(1/n^\nu)$ with $\nu>1/4$, then the
			one-step GN refinement $ \widehat{\theta}_n^{gn}$ is  asymptotically equivalent to the PEM, namely,
			$\sqrt{n}(\widehat{\theta}_n^{gn} - \widehat{\theta}^{\rm pem}_{n})=o_p(1).$

		\end{lem}
		An appealing  advantage of the two-step estimator described above is that it only requires a consistent initial estimator satisfying $\widehat{\theta}_n-\theta^o=O_p(1/n^\nu),~\nu>1/4$ rather than the more demanding task of carefully designing an estimator and analyzing its exact rate of convergence.
		The next two sections are devoted to achieving this objective.
	}

\section{Consistent estimators of two auxiliary submodels}
\label{sec3}
In this section, we make a theoretical preparation for deriving a  consistent estimator of  the BJ model \eqref{gm}, which includes the  consistent estimators of   the autoregressive with exogenous input model of infinite order (ARX($\infty$) model) and the    OE  model.

\subsection{Consistent estimators of ARX($\infty$) models}
\label{sec3.1}
Consider the ARX ($\infty$) model  described in \cite{Ljung1992} by
\begingroup
\allowdisplaybreaks
\begin{subequations}\label{arx}
	\begin{align}
		 & V^o(q)y(t)=W^o(q)u(t)+e(t), ~t=1,\cdots,n                                       \\
		 & V^o(q)=1+\sum_{k=1}^\infty v_k^o q^{-1},~W^o(q)=\sum_{k=1}^\infty w_k^o q^{-1},
	\end{align}
\end{subequations}
\endgroup
where  $\sum_{k=1}^\infty \sqrt{k}|v_k^o|<\infty$ and $\sum_{k=1}^\infty \sqrt{k}|w_k^o|<\infty$.
By adopting the techniques in \cite{Ljung1992}, the model \eqref{arx}  is approximated by a high-order ARX model
\begin{subequations}
	\label{harx}
	\begin{align}
		 & V(q)y(t)=W(q)u(t)+e(t),                                       \\
		 & V(q)=1+\sum_{k=1}^m v_k q^{-k},~W(q)=\sum_{k=1}^m w_k q^{-k},
	\end{align}
\end{subequations}
where $m$ is the order of the approximate model \eqref{harx} and is a function of $n$.
	{\color{blue}The discussion on how to select \(m\) is postponed to Section \ref{sec5.2}.}
Denote the first $m$ parameters of the true polynomials $V^o(q)$ and $W^o(q)$ by
\begin{align*}
	\theta_{vw}^o
	\!=[\theta_{v}^{oT},\theta_{w}^{oT}]^T,
	\theta_{v}^o\!=[v_1^o,\cdots,v_m^o]^T,
	\theta_{w}^o\!=[w_1^o,\cdots,w_m^o]^T.
\end{align*}
Here, we aim to use the data  $\{u(t),y(t),t=1,\cdots,n\}$ generated by the model \eqref{arx} to estimate the $2m$ parameters $\theta_{vw}^o$ in terms of the truncated model \eqref{harx}.
By letting $\theta_{vw}=
	[\theta_v^T,\theta_w^T]^T$ with $\theta_v=[v_1,\cdots,v_m]^T$ and
$\theta_w=[w_1,\cdots,w_m]^T$,
the ARX model \eqref{harx} has a linear regression form:
\begin{subequations}\label{arx2}
	\begin{align}
		y & =X\theta_{vw} + e,           \\
		y & =[y(1),y(2),\cdots,y(n)]^T,  \\
		X & = [x(1),x(2),\cdots,x(n)]^T, \\
		e & = [e(1),e(2),\cdots,e(n)]^T,
	\end{align}
\end{subequations}
and regressor $x(t) = [-y(t-1),\cdots,-y(t-m),u(t-1),\cdots,u(t-m)]^T$.
As a result, the parameters $\theta_{vw}^o$ are estimated by the LS method
\begin{align}
	\label{ls}
	\widehat{\theta}_n^{vw}=[\widehat{\theta}_n^{vT},
	\widehat{\theta}_n^{wT}]^T
	\eq(X^TX)^{-1}X^Ty.
\end{align}
Accordingly, we denote the estimators for $V^o(q)$ and $W^o(q)$ by
	\begin{align*}
		   V(q,\widehat{\theta}_n^v)
		=1 + \sum_{k=1}^m \widehat{v}_kq^{-k},~~ W(q,\widehat{\theta}_n^w)
		=1 + \sum_{k=1}^m \widehat{w}_kq^{-k},
	\end{align*}  
where $\widehat{v}_k$ and $\widehat{w}_k$
are the $k$-th entry of
$\widehat{\theta}_n^{v}$ and $\widehat{\theta}_n^{w}$, respectively.
We have the following convergence results on the estimator \eqref{ls}.
{\color{red!70!black}
\begin{lem}
	\label{lm2}
	\cite[Theorem 6.1 and Lemma 5.1]{Ljung1992}
	Consider the ARX($\infty$) model \eqref{arx}.
	Suppose Assumptions \ref{assum3} and \ref{assum4}  hold for the ARX($\infty$) model and further suppose  the truncated model order $m$ satisfies
	\begin{enumerate}[(i)]
		\item $m\xra{}\infty$ as $n\xra{}\infty$;
		\item $m^{3+\kappa}/n\xra{}0$ as $n\xra{}\infty$ for some $\kappa>0$.
	\end{enumerate}
	Thus, the estimator  $\widehat{\theta}_n^{vw}$ converges to $\theta_{vw}^o$  in probability as $n\xra{}\infty$
	with the rate of convergence
		\begin{align} 	\label{arx_rate} \big\|\widehat{\theta}_n^{vw}-\theta_{vw}^o\big\|_1
			=O_p(\delta_n),~~\delta_n\eq \frac{m}{\sqrt{ n}}+d_m,   ~~
			d_m\eq\sum_{k=m+1}^\infty |v_k^o| + |w_k^o|.
		\end{align} 
\end{lem}
\textbf{Proof.} See proof in Appendix B.
}

\subsection{Consistent estimators of OE models}
\label{sec3.3}
Consider the OE model described by
\begin{align}
	y(t)= \frac{B(q)}{F(q)}u(t) +e(t),\label{oe}
\end{align}
which is a special case of the BJ model \eqref{gm} with $C(q)\equiv D(q)\equiv 1$.
Denote the  true parameters $\theta_{fb}^o=[\theta_{f}^{oT},\theta_{b}^{oT}]^T$  of the OE model \eqref{oe} with $\theta_{f}^{o}=[f_1^o,f_{2}^o,\cdots,f_{p_f}^o]^T$ and $\theta_{b}^{o}=[b_1^o,b_{2}^o,\cdots,b_{p_b}^o]^T$.

The estimator obtained by directly applying the LS method to the OE model using data $\{u(t), y(t),t=1,\cdots,n\}$ is biased. Consider transforming the OE model into a specific form of regression.
Let
\begin{align*}
	y^o(t) \eq \frac{B^o(q)}{F^o(q)}u(t)
\end{align*}
be noise-free output of the OE model, yielding the identity
\begin{align*}
	y^o(t) = \phi(t)^T \theta_{fb}^o
\end{align*}
with   $\phi(t)=[-y^o(t-1),\cdots,-y^o(t-p_f),u(t-1),\cdots,u(t-p_b)]^T$.
Note that $y(t)=y(t)^o +e(t)$.
Consequently,
we can write the OE model \eqref{oe} as the equivalent linear regression model
\begin{align}
	y(t) = \phi(t)^T \theta_{fb}^o +e(t).\label{lmoe}
\end{align}
It can be verified that the  LS estimator of the  model \eqref{lmoe} is unbiased and consistent. However, we need to use the noise-free output $y^o(t)$, which is unobservable. If we can obtain a consistent estimate of the data used in \eqref{lmoe}, then the LS estimator will still guarantee consistency.

Given the estimates $\{\widehat{u}(t), \widehat{y}^o(t), \widehat{y}(t)\}$ of $\{u(t), y^o(t), y(t)\}$ for $t = 1, \ldots, n$ in the OE model \eqref{lmoe}, we define
\begingroup
\allowdisplaybreaks
	\begin{align}
		 & \widehat{y}\eq [\widehat{y}(1),\widehat{y}(2),\cdots,\widehat{y}(n)]^T,              \\
		 & \widehat{\Phi} \eq  [\widehat{\phi}(1),\widehat{\phi}(2),\cdots,\widehat{\phi}(n)]^T,\label{xy} 
	\end{align}
where $\widehat{\phi}(t) \eq  [ -\widehat{y}^o(t-1), \cdots, -\widehat{y}^o(t-p_f), \widehat{u}(t-1), \cdots, $
		$\widehat{u}(t-p_b) ]^T$. Thus, the LS estimator for $\theta_{fb}^o$ is given by
\begin{align} 	\label{cls}
	 & \widehat{\theta}_n^{fb}
	\eq  ( \widehat{\Phi}^T\widehat{\Phi}
	)^{-1} \widehat{\Phi}^T\widehat{y} .
\end{align}
The above discussion is summarized in the following lemma:

\begin{lem}
	\label{lm3}
	Consider the OE model \eqref{oe}.
	Suppose that the following assumptions hold:
	\begin{enumerate}[(i)]
		\item The polynomials $q^{p_b}B^o(q)$ and $q^{p_f}F^o(q)$ have no common factor, and moreover $F^o(q)$ is stable.
		\item The noise sequence $\{e(t)\}$ is a
		      stochastic process that satisfies
		      $E(e(t)|\mathscr{F}_{t-1})=0$ with $\mathscr{F}_t\eq\sigma\{e(s),u(s),0\leq s\leq t\}$, $E(e(t)^2)=\sigma^2$ and $E(e(t)^4)<\infty$.

		\item The input sequence $\{u(t)\}$ is  persistently exciting of order $p_f+p_b$ and {\color{blue}the regressor $\{\phi(t)\}$ is uncorrelated with the noise sequence $\{e(t)\}$.}

		\item The estimates for the inputs and outputs satisfy 
			      \begin{align}
				        |\widehat{u}(t) -u(t)|=O_p(\zeta_n),                   ~~
				         {\color{blue}|\widehat{y}^o(t) -y^o(t)|=O_p(\zeta_n)	,} ~~ |\widehat{y}(t) -y(t)|=O_p(\zeta_n)
			      \end{align}
		      for all $t=1,\cdots,n$, where the deterministic sequence $\zeta_n\xra{}0$ as $n\xra{}\infty$.
	\end{enumerate}
	{\color{red!70!black}Thus,
	the LS estimator $\widehat{\theta}_n^{fb}$  defined by \eqref{cls} converges to its true value $\theta_{fb}^o$ in probability as $n\xra{}\infty$  with the rate of convergence} $$\|\widehat{\theta}_n^{fb}-\theta_{fb}^o\|_2=\max\{O_p(\zeta_n),O_p(1/\sqrt{n})\}.$$

\end{lem}
{\color{red!70!black} \text{Proof.} See proof in Appendix B.}

\section{Sequentially decoupling estimators and its GN refinement}
\label{sec4}

In this section, we develop the  SD estimator {\color{blue}for  both open-loop and closed-loop settings}, which progressively separates the dynamic and noise components of the  BJ model \eqref{gm}, and establish its consistency with a convergence rate of {\color{blue}\( O_p(\delta_n) \) in probability, where \( \delta_n \) is defined in \eqref{arx_rate}}, {\color{blue}and prove the asymptotic efficiency of the proposed SDGN method.}

\subsection{Algorithm of SD estimators}
\label{4.2}
In this subsection, we present the   algorithm for deriving the SD  estimator to successively estimate the four polynomials using the input-output data. The approach decouple the BJ model \eqref{gm} by successively estimating two auxiliary OE models \eqref{lmoe} with respect to the parameters of interest using filtered  data, thereby deriving a consistent estimator of the four polynomial parameters.

The algorithm is as follows:
\begin{enumerate}[(i)]
	\item Estimate the parameters $\theta_{v}^o$ {\color{blue}and $\theta_{w}^o$} (the first $m$ parameters of $V^o(q)=D^o(q)/C^o(q)$ and {\color{blue}$W^o(q)=B^o(q)D^o(q)/(F^o(q)C^o(q))$})
	      of the ARX($\infty$) model
	      \begin{align}
		      \underbrace{\frac{D^o(q)}{C^o(q)}}_{V^o(q)}y(t)=\underbrace{\frac{D^o(q)}{C^o(q)}\frac{B^o(q)}{F^o(q)}}_{W^o(q)}u(t)+e(t)\label{sea1}
	      \end{align}
	      using the approximate high-order ARX model \eqref{harx} of order $m$
	      \begin{align}\label{fsea1}
		      V(q)y(t)=W(q)u(t)+e(t)
	      \end{align}
	      in terms of the data $\{u(t),y(t),t=1,\cdots,n\}$  by the LS estimator \eqref{ls} given in Section \ref{sec3.1}.
	      Denote the estimated parameters for $V(q)$ {\color{blue}and $W(q)$} of the model \eqref{fsea1} by $ \widehat{\theta}_n^v$ {\color{blue}and $ \widehat{\theta}_n^w$}
	      as well as its corresponding polynomials
	      $V(q,\widehat{\theta}_n^{v})$ {\color{blue} and $W(q,\widehat{\theta}_n^{w})$.}

	\item

	      Estimate the parameters $\theta_{b}^{o}$ and $\theta_{f}^{o}$ of the OE model
	      \begin{align}
		      \underbrace{\frac{D^o(q)}{C^o(q)}y(t)}_{y_V^f(t)}=\overbrace{\frac{B^o(q)}{F^o(q)}\underbrace{\frac{D^o(q)}{C^o(q)}u(t)}_{u_V^f(t)}}^{{\color{blue}y_V^{of}(t)}} + e(t)
		      \label{sea2}
	      \end{align}
	      in terms of the estimated filtered signals
	      \begin{align*}
		      \big\{
		      \widehat{u}_V^f(t)\eq V(q,\widehat{\theta}_n^{v})u(t),~
		       & {\color{blue}\widehat{y}_V^{of}(t)\eq W(q,\widehat{\theta}_n^{w})u(t),}~ \widehat{y}_V^f(t)\eq V(q,\widehat{\theta}_n^{v})y(t)
		      \big\}
	      \end{align*}
	      for $t=1,\cdots,n$ 	by the LS estimator \eqref{cls} given in Section \ref{sec3.3}.
	      Denote the estimated parameters of $B^o(q)$ and $F^o(q)$ by $ \widehat{\theta}_n^b$ and $ \widehat{\theta}_n^{f}$
	      as well as their corresponding polynomials by
	      $B(q,\widehat{\theta}_n^{b})$ and $F(q,\widehat{\theta}_n^{f})$.

	\item 	Estimate the parameters $\theta_{c}^{o}$ and $\theta_{d}^{o}$ of the OE model \begin{align}
		      \underbrace{\frac{D^o(q)}{C^o(q)}y(t)}_{y_V^f(t)}=\overbrace{\frac{D^o(q)}{C^o(q)}\underbrace{\frac{B^o(q)}{F^o(q)}u(t)}_{u_{BF}^f(t)}}^{{\color{blue}y_V^{of}(t)}}+e(t)\label{sea3}
	      \end{align}
	      in terms of the estimated filtered signals
	      \begingroup
	      \allowdisplaybreaks
	      \begin{align*}
		      \big\{ 	\widehat{u}_{BF}^f(t)\eq
		      \frac {B(q,\widehat{\theta}^b_n)}{F(q,\widehat{\theta}^f_n)}u(t),~
		        {\color{blue}\widehat{y}_V^{of}(t)\eq W(q,\widehat{\theta}_n^{w})u(t),}  ~\widehat{y}_V^f(t)\eq V(q,\widehat{\theta}_n^{v})y(t)\big\}
	      \end{align*}
	      \endgroup
	      for $t=1,\cdots,n$ 	by the LS estimator  \eqref{cls} given in Section \ref{sec3.3}.
	      Denote the estimated parameters of $C^o(q)$ and $D^o(q)$ by $ \widehat{\theta}_n^c$ and $ \widehat{\theta}_n^{d}$, respectively.

	\item Get the SD estimator $\widehat{\theta}_n^{\rm sd}$ for the true parameters $\theta^o$  by stacking the estimates in the way $[\widehat{\theta}^{bT}_n,\widehat{\theta}^{cT}_n,\widehat{\theta}^{dT}_n,\widehat{\theta}^{fT}_n]^T$.
\end{enumerate}

\begin{rem}
	{\color{blue}The algorithm for the SD estimator mainly involves three standard least squares  and four filtered signals.}
	It avoids costly iterative optimization, making the SD estimator highly efficient in practice.
\end{rem}
{\color{blue}
\begin{rem}

	Note that the leading coefficient \(d_0 = 1\) of \(D^o(q)\) in the OE model \eqref{sea3} is implicitly accounted for in the  LS estimator \eqref{cls}.
	Indeed, by moving the signal \(u_{BF}^f(t)\) to the left-hand side, we can rewrite the OE model \eqref{sea3} as the following linear regression form:
	\begin{align}
		y_V^f(t) - u_{BF}^f(t) = \phi(t)^\top \theta_{cd}^o + e(t), \label{secdd}
	\end{align}
	where $\phi(t)=[-y_V^{of}(t-1),\cdots,-y_V^{of}(t-p_f),u_{BF}^f(t-1),\cdots,u_{BF}^f(t-p_b)]^T$ and $\theta_{cd}^o=[\theta_{c}^{oT},\theta_{d}^{oT}]^T$.
	Consequently, the parameter vectors \(\theta_c^o\) and \(\theta_d^o\) can be directly estimated from \eqref{secdd} using the estimated filtered signals \(\widehat{u}_{BF}^f(t)\), \(\widehat{y}_V^{of}(t)\), and \(\widehat{y}_V^f(t)\).

\end{rem}

\begin{rem}
	The consistent estimates \(\widehat{\theta}_n^v\), \(\widehat{\theta}_n^w\), \(\widehat{\theta}_n^b\), and \(\widehat{\theta}_n^f\) serve as filter coefficients that yield consistent approximations of the  unobservable signals \(\{\widehat{u}_V^f(t), \widehat{u}_{BF}^f(t),$ $ \widehat{y}_V^{of}(t), \widehat{y}_V^f(t)\}\), which together form two OE model structures used to decouple and estimate the dynamic and noise components of the BJ model.
\end{rem}
}

\subsection{Consistency and {\color{blue}asymptotic efficiency}}
\label{sec4.2}
In this subsection, we aim to establish the consistency and rate of convergence in probability for the SD estimator based on the convergence results for the ARX($\infty$) and OE models introduced in Section \ref{sec3}, and to prove the {\color{blue}asymptotic efficiency of the SDGN method}.

{\color{blue}
Before presenting the rate of convergence, we introduce a constant $\rho$ associated with the polynomials $C^o(q)$ and $F^o(q)$. Let $\{\rho_i,i=1,\cdots, p_c+p_f\}$ denote all the roots of the polynomials $C^o(q)$ and $F^o(q)$ and define
\begin{align}
	\rho \eq \max_{1 \leq i \leq p_c+p_f} |\rho_i|.\label{rho}
\end{align}}
\begin{thm}
	\label{thm1}
	Suppose that Assumptions \ref{assum1}--\ref{assum4} hold.
	Moreover, let the truncation order $m$   for  the ARX($\infty$) model \eqref{sea1} satisfies
	\begin{enumerate}[(i)]
		\item $m\xra{}\infty$ as $n\xra{}\infty$;
		\item {\color{blue}$m^{3+\kappa}/n\xra{}0$ as $n\xra{}\infty$ for some $\kappa>0$.}
	\end{enumerate}
	Thus, the SD estimator is consistent with {\color{blue} the rate of convergence in probability: }
	$
		\|\widehat{\theta}_n^{\rm sd} - \theta^o\|_2 = O_p(\delta_n).
	$
\end{thm}
{\color{red!70!black} \text{Proof.} See proof in Appendix A.

Theorem \ref{thm1} demonstrates  that the  rate of convergence in probability of the SD  estimator is $O_p(\delta_n)$, which depends on the truncation order $m$ of the ARX($\infty$) model \eqref{sea1}.
The term $O_p(m/\sqrt{n})$ represents the square root of the variance, which increases monotonically as $m$ increases. On the other hand,  the term $O(\rho^m)$ represents the model approximation bias when using the ARX model \eqref{fsea1} of order $m$ to approximate the ARX($\infty$) model \eqref{sea1}, and this bias decreases monotonically as $m$ increases.
Therefore,  the SD will achieves its fastest rate when both $O_p(m/\sqrt{n})$ and $O(\rho^m)$ are of   the same   order.}

	{\color{blue}
		The following proposition further specifies the attainable rate  of convergence in probability of the SD estimator for typical choices of $m$, and identifies the fastest achievable rate.}
\begin{prop}
	\label{thm2}
	Suppose that Assumptions \ref{assum1}--\ref{assum5} hold.
	We have  the rate of convergence  of the SD  estimator  as follows:
	\begin{enumerate}[(i)]
		\item when $m=O(n^\tau)$ for any $0<\tau<1/2$, we have
		      \begin{align*}
			      \delta_n=O\Big(\frac1{n^{\frac12-\tau}}\Big),~~
			      {\color{blue}\|\widehat{\theta}_n^{\rm sd} - \theta^o\|_2 = O_p\Big(\frac1{n^{\frac12-\tau}}\Big)}
		      \end{align*}
		      since $\frac{m}{\sqrt{n}}=O\Big(\frac1{n^{\frac12-\tau}}\Big),~\rho^m=O(\rho^{n^\tau}),$ and $n^{\frac12-\tau}\rho^{n^\tau}\xra{}0$ as $n\xra{}\infty$;
		\item when $m=\alpha \log n$ with $\alpha=-1/(2\log \rho)>0$, we have
		      \begin{align*}
			      \delta_n=O\Big(\frac{\log n}{\sqrt{n}}\Big),~~{\color{blue}\|\widehat{\theta}_n^{\rm sd} - \theta^o\|_2 = O_p\Big(\frac{\log n}{\sqrt{n}}\Big)}
		      \end{align*}
		      since $\frac{m}{\sqrt{n}}=\frac{\alpha\log n}{\sqrt{n}}$ and $\rho^m=\frac{1}{\sqrt{n}}$;

		\item when $m=\alpha \log n - 2\alpha \log\log n$ with $\alpha=-1/(2\log \rho)>0$, we have
		      \begin{align*}
			      \delta_n=O\Big(\frac{\log n}{\sqrt{n}}\Big),~~{\color{blue}\|\widehat{\theta}_n^{\rm sd} - \theta^o\|_2 = O_p\Big(\frac{\log n}{\sqrt{n}}\Big)}
		      \end{align*}
		      since both $\frac{m}{\sqrt{n}}=\frac{\alpha\log n- 2\alpha \log\log n}{\sqrt{n}}$ and $\rho^m=\frac{\log n}{\sqrt{n}}$ achieve the same order.
	\end{enumerate}
\end{prop}
{\color{red!70!black} \text{Proof.} See proof in Appendix A.}

We can summarize the following insights on the convergence rate of the SD  estimator  from {\color{blue}Proposition}  \ref{thm2}:
\begin{enumerate}[(i)]
	\item{\color{blue} The SD estimator is consistent for all three typical choices of the truncation order $m$ considered above, that is, $\|\widehat{\theta}_n^{\rm sd} - \theta^o\|_2 = o_p(1)$.

	\item  By examining the truncation orders successively from (i) to (iii),  } we find that the fastest achievable convergence rate is \( O_p\left( \log n/\sqrt{n}\right) \), {\color{blue} which is slightly slower than the standard parametric rate  \( O_p\left(1/\sqrt{n}\right) \).} This fastest rate arises when both  terms are of  the same   order.


	\item  Suppose the estimation error satisfies the upper bound \( \|\widehat{\theta}_n^{\rm sd} - \theta^o\|_2 \leq C_1\frac{m}{\sqrt{n}} + C_2\rho^m \), where \( C_1 \) and \( C_2 \) are   leading  constants. {\color{blue} Then among all choices of \(m\) that yield the optimal order
			      \(O_p\!\left( \log n/\sqrt{n}\right)\), the bound
			      $
				      C_1  \alpha \log n/\sqrt{n}
			      $
			      is asymptotically smallest (i.e., optimal in terms of the leading constant) when \(m = \alpha \log n\) with \(\alpha = -1/(2\log \rho) > 0\). This choice balances the two error sources while minimizing the dominant term in the upper bound.}
\end{enumerate}
{\color{blue}
Based on the rate presented in Proposition  \ref{thm2} and Lemma \ref{lm1}, we have the following results on the proposed SDGN method.
\begin{thm}
	\label{thm3}
	Suppose that Assumptions \ref{assum1}--\ref{assum5} hold and  the loss function $\ell_n(\theta)$ is three-times  differentiable  and $\big|\frac{\partial^3 \ell_n(\theta)}{\partial \theta_i\partial \theta_j\partial \theta_k} \big|$ are upper bounded by a uniform constant on a compact set of $\theta^o$.
	Let $\widehat{\theta}_n^{\rm{sdgn}}$ be the one-step GN refinement \eqref{gn} with $\widehat{\theta}_n$ replaced by   the SD estimator $\widehat{\theta}_n^{\rm{sd}}$.
	Then, $\widehat{\theta}_n^{\rm sdgn}$
	is asymptotically equivalent to the PEM if the truncation order $m$ asymptotically satisfies
	\begin{align*}
		-1/(2\log \rho) \log(n) - a_n \leq  m \leq C n^\gamma,
	\end{align*}
	where $a_n$ is any positive sequence satisfying $a_n=o(\log(n))$, $\rho$ is given by \eqref{rho}, $ \gamma$ is any constant satisfying $0< \gamma<1/4$, and $C$ is any positive constant.
\end{thm}
}{\color{red!70!black}
In practice, we can run the GN iterations until it converges for finite sample size.
Theorem \ref{thm1} guarantees that the SD estimator is consistent and converges to the true value at a rate of $O_p(\delta_n)$, so it lies within a small neighborhood of the true value. The GN iteration will stop after only a few steps.

}

\section{Numerical illustrations}
\label{sec5}
In this section, we conduct Monte Carlo simulations to evaluate the numerical performance of the proposed SD estimator and its  GN refinement SDGN  for  BJ model estimation. The results are compared with existing PEM, WNSF, MORSM, BJSM, and RIV methods.

{\color{blue}
\subsection{Asymptotic efficiency in both open-loop and closed-loop scenarios }
\label{sec5.1}
This example is to show that the proposed SD is consistent and  SDGN is asympototically efficient in  both  open-loop and closed-loop scenarios.
We use the same simulation settings given  in \cite[Section V.A]{Galrinho2019J1}.
Consider the BJ model \[\quad y(t) = \underbrace{\frac{q^{-1} + 0.1q^{-2}}{1 - 0.5q^{-1} + 0.75q^{-2}}}_{G^o(q)}u(t) + \underbrace{\frac{1 + 0.7q^{-1}}{1 - 0.9q^{-1}}}_{H^o(q)}e(t),
\] where \(\{e(t)\}\) is an independent Gaussian white sequences with unit variance.
For the open-loop data, the input  is generated by
$
	u(t) = \frac{1}{1 +  G^o(q)}r(t),
$
where \(\{r(t)\}\) is an independent Gaussian white sequences with unit variance.
For the closed-loop data, the input is
$
	u(t) = - y(t) + r(t)
$  and \(\{r(t)\}\) is an independent Gaussian white sequences with unit variance.

\begin{figure}[ht]
	\centering
	\includegraphics[scale=0.5]{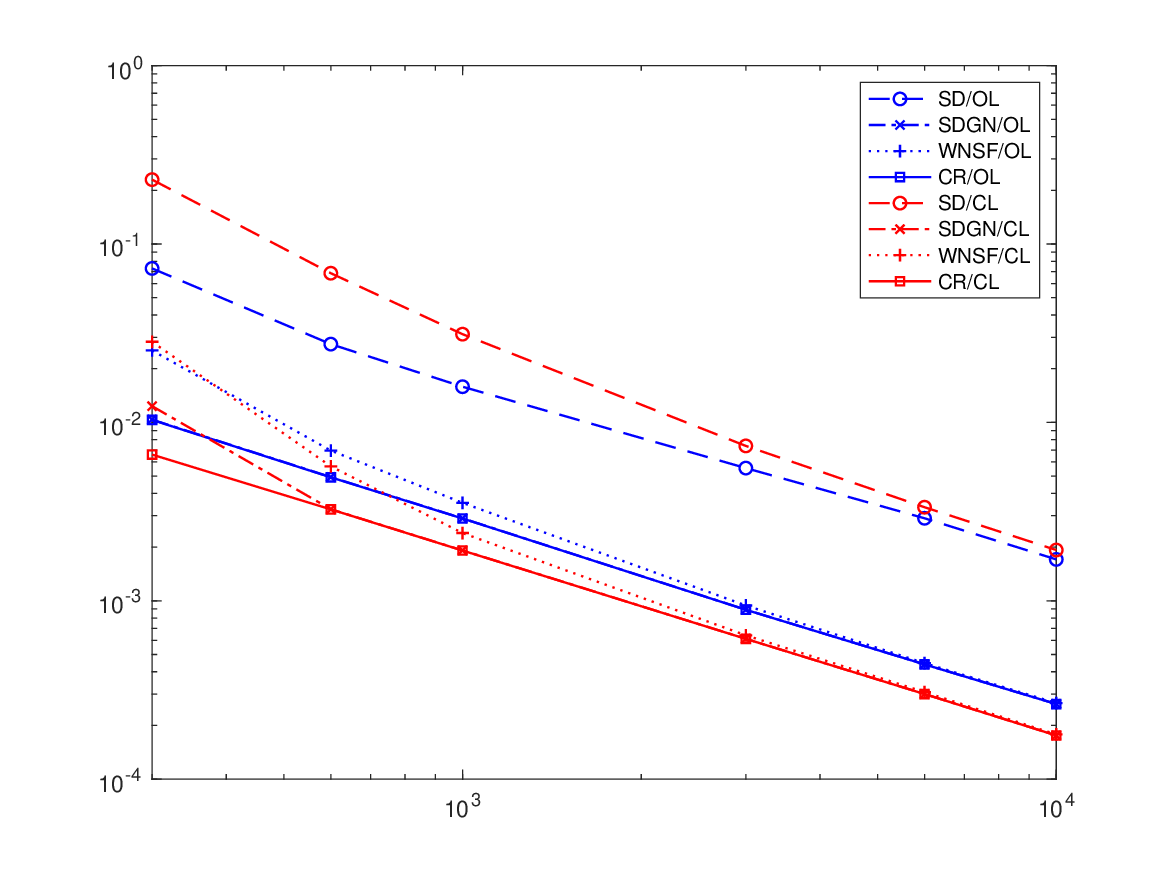}
	\caption{The average MSEs over 1000 Monte Carlo runs.}
	\label{mse1}
\end{figure}

We perform 1000 Monte Carlo runs for both open-loop and closed-loop data with sample sizes \(n=300, 600, 1000, 3000, 6000, 10000\) under zero initial conditions.
We illustrate the performance of the estimators: SD, SDGN, and WNSF.
We set the order of  the ARX model  involved both in the SD and WNSF estimators to be 50 for  open- and closed-loop data as used in \cite{Galrinho2019J1}.
We evaluate the performance of the three estimators by  the mean-squared error of the estimated parameter vector of the dynamic model, MSE = \(||\hat{\bar{\theta}}_n^{\rm fb} -  \theta^o_{fb}||^2\), where \(\hat{\bar{\theta}}_n^{\rm fb}\) is the corresponding estimate and $\theta^o_{fb}$ is the true value.

We present the average MSEs over 1000 Monte Carlo runs in Fig. \ref{mse1} as  function of   sample size, where the OL an CL denotes the open-loop and closed-loop for brevity.
Note that for open-loop data, the dash-dotted line with stars (SDGN) coincides with the solid line with squares (Cramér-Rao (CR) lower bound).
We find that the SD estimator is consistent as the sample size increases and further the SDGN   is asympotically efficient for both the open-loop and closed-loop data.
As illustrated in \cite{Galrinho2019J1},
the WNSF estimator is also asympotically efficient for both the open-loop and closed-loop data.
}

\subsection{Oscillatory BJ model estimation under low-pass open-loop excitation}
\label{sec5.2}
In this subsection, we illustrate  {\color{blue}
		the numerical performance of SD,  WNSF, MORSM, BJSM, and RIV methods and their GN refinements for a strong oscillatory BJ model estimation  under low-pass open-loop excitation.
	}

\subsubsection{Simulation settings}
We consider the   BJ model with strong oscillation as described in \cite{Zhu2016}:
\begin{align}
	y(t)  =\frac{q^{-1}+0.5q^{-2}-2q^{-3}+q^{-4}}{1-1.5q^{-1}+0.7q^{-2}+0.3q^{-3}-0.2q^{-4}}u(t) +\frac{1-0.6q^{-1}+0.4q^{-2}}{1-1.95q^{-1}+0.9506q^{-2}}e(t).
	\label{bjnum}
\end{align}
{\color{blue}
The input is an iid Gaussian random sequence with zero mean and unit variance filtered by the transfer function $1/(1-0.85q^{-1})^2$.
For each input realization \( u(t) \), the output \( y(t) \) is simulated using the BJ model \eqref{bjnum}, driven by the input \( u(t) \) and an iid Gaussian white noise sequence \( e(t) \). The variance of \( e(t) \) is chosen so that  the ratio of the squared sum  between the noise-free output  and the noise  $e(t) $ equals 3. We generate 500 independent realizations, each of length 20000.
To assess how   estimator  performance varies with sample size, we report simulation results for \( n = 2500, 5000, 10000, \) and \( 20000 \).
}

\subsubsection{Estimators}
\label{est}

We compare our estimators
\begin{itemize}
	\item SD: The SD estimator following the algorithm   described in Section \ref{4.2};
	\item SDGN: The estimator obtained through GN iteration using the SD estimator as its initial value;
\end{itemize}
with
the following estimators:
\begin{itemize}
	\item PEMd: The PEM   initialized by the default  value \cite{Ljung2012};

	      {\color{blue} \item PEMt: The PEM   initialized by the true value;

	\item WNSF: The WNSF estimator  developed in \cite{Galrinho2019J1};

	\item PEMw: The PEM   initialized by the WNSF estimator;

	\item MORSM: The MORSM estimator following the method proposed in \cite{Everitt2018} with one iteration;

	\item PEMm: The PEM   initialized by the MORSM estimator;}

	\item BJSM: The BJSM estimator  implemented by following the  settings outlined in \cite[Section 5]{Zhu2016};

	\item PEMb: The PEM  initialized by the BJSM estimator.

	\item RIV: The RIV estimator developed in \cite{Young2015}, implemented using the command {\ttfamily rivbj} in the CAPTAIN Toolbox for MATLAB\footnote{The CAPTAIN Toolbox can be downloaded   from https://wp.lancs.ac.uk/captaintoolbox.}.

	\item PEMr: The PEM  initialized by the RIV estimator.

\end{itemize}

All PEM-based estimators, including SDGN, PEMd, PEMt, PEMw, PEMm, PEMb, and PEMr, are implemented using MATLAB’s System Identification Toolbox via   the {\ttfamily bj} command  with the SearchMethod option set to `gn' in MATLAB's System Identification Toolbox \cite{Ljung2012}. Each estimator is initialized with its corresponding initial   estimator.
The stopping criteria for all the PEM-based estimators  and the BJSM estimator are set to  a maximum of 100 iterations unless the tolerance reaches $10^{-4}$  \cite{Galrinho2019J1}.

All the computations were executed on a MacBook Air equipped with an Apple M2 chip and 24GB RAM under the Matlab 2023b platform.
\begin{table*}[h]
	\centering
	{\color{blue}
		\caption{The average   fits of all estimators   among 500 realizations  under different sample sizes.}
		\label{fit1}
		\resizebox{\textwidth}{!}{%
		\begin{tabular}{ccccclcccrcccc}
			\hline
			$n$         & SD    & SDGN & WNSF & PEMw & MORSM & PEMm & BJSM & PEMb & RIV & PEMr & PEMd & PEMt \\ \hline
			2500        & 34.51 &
			55.29       &
			20.59       &
			41.04       &
			32.77 (0)   &
			52.20       &
			30.27       &
			51.97       &
			-88.32      &
			-77.77      &
			46.74       &
			69.52                                                                                            \\ \hline
			5000        &
			37.64       &
			71.33       &
			21.91       &
			58.59       &
			39.51 (92)  &
			68.79       &
			34.81       &
			71.02       &
			-79.92      &
			-71.53      &
			59.96       &
			80.53                                                                                            \\ \hline
			10000       &
			40.43       &
			82.36       &
			24.46       &
			69.44       &
			50.78 (197) &
			83.64       &
			36.16       &
			83.77       &
			-119.83     &
			-73.99      &
			63.87       &
			86.86       &                                                                                    \\ \hline
			20000       &
			43.68       &
			90.74       &
			26.51       &
			75.19       &
			61.16 (234) &
			90.27       &
			37.49       &
			90.66       &
			-65.96      &
			-59.58      &
			67.95       &
			91.06       &                                                                                    \\ \hline
		\end{tabular}}
	}

\end{table*}

\begin{table}[H]
	\centering
	{\color{blue}
		\caption{The average running times of the estimators  without using GN refinement  among 500 realizations under different sample sizes  (Unit: Seconds).}
		\label{time}
		\begin{tabular}{cccccc}
			\hline
			$n$    & SD     & WNSF & MORSM & BJSM & RIV \\ \hline
			2500   & 0.0044 &
			0.0051 &
			0.0069 &
			0.0723 &
			1.8308                                      \\ \hline
			5000   &
			0.0061 &
			0.0069 &
			0.0090 &
			0.0905 &
			3.7309                                      \\ \hline
			10000  &
			0.0101 &
			0.0113 &
			0.0133 &
			0.1269 &
			9.1367                                      \\ \hline
			20000  &
			0.0177 &
			0.0206 &
			0.0222 &
			0.3594 &
			19.4684                                     \\ \hline
		\end{tabular}
	}
\end{table}

\begin{table}[H]
	\centering
	{\color{blue}
		\caption{The average number of iterations of the estimators using GN refinement  among 500 realizations  under different sample sizes.}
		\label{max_iter}
		\vspace{1.5ex}
		\begin{tabu} to 0.7
			\textwidth{cccccccc}
			\hline
			$n$   & SDGN & PEMw & PEMm & PEMb & PEMr & PEMd & PEMt \\ \hline
			2500  &
			10.08 &
			10.58 &
			11.52 &
			9.20  &
			12.72 &
			27.26 &
			7.93
			\\ \hline
			5000  &
			7.89  &
			8.57  &
			8.96  &
			7.43  &
			10.88 &
			29.89 &
			5.87                                                   \\ \hline
			10000 &
			6.47  &
			7.08  &
			6.62  &
			6.05  &
			8.58  &
			28.73 &
			4.24                                                   \\ \hline
			20000 &
			5.89  &
			6.58  &
			5.15  &
			5.34  &
			7.01  &
			26.97 &
			3.33                                                   \\ \hline
		\end{tabu}}
\end{table}

\begin{figure}[ht]
	\centering
	\begin{subfigure}[b]{0.49\linewidth}
		\includegraphics[width=\linewidth]{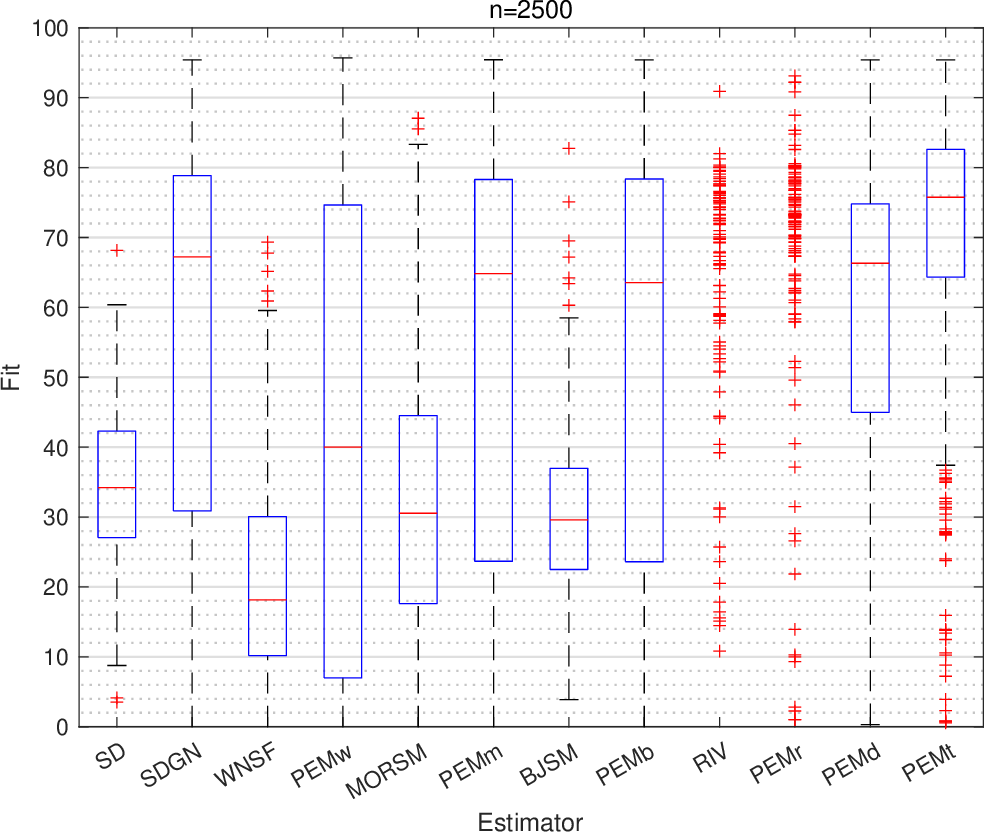}
		\caption{$n=2500$}
		\label{fit_2500_aic}
	\end{subfigure}\hfill
	\begin{subfigure}[b]{0.49\linewidth}
		\includegraphics[width=\linewidth]{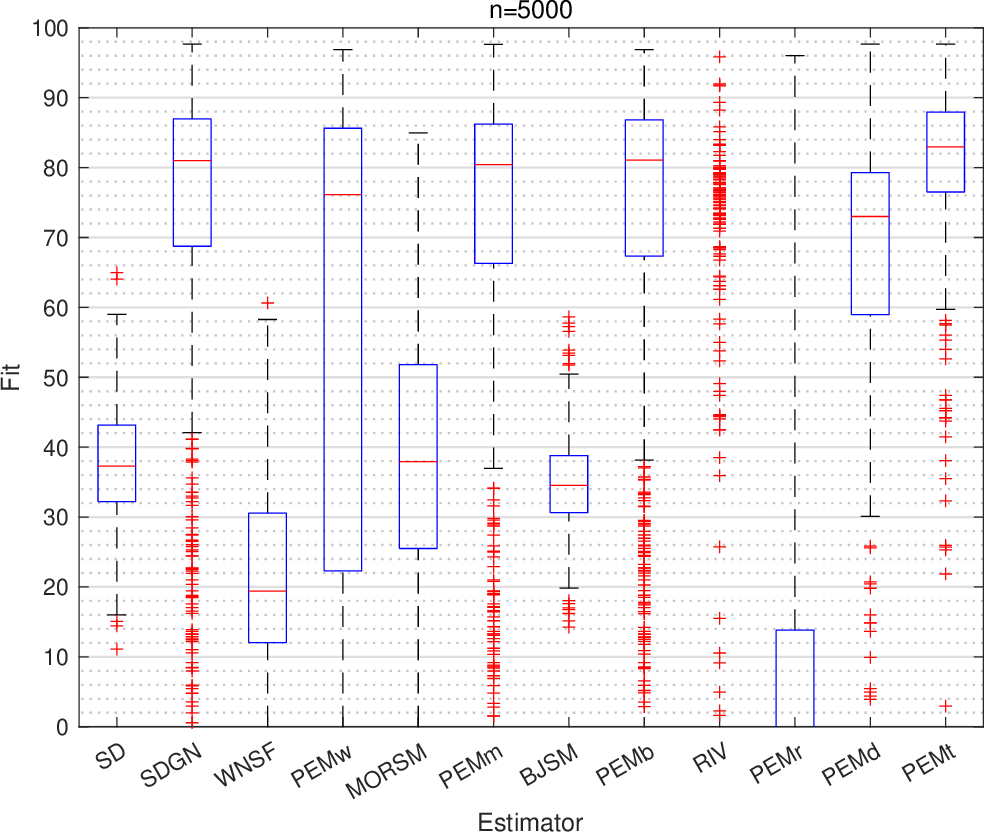}
		\caption{$n=5000$}
		\label{fit_5000_aic}
	\end{subfigure}

	\vspace{2mm}

	\begin{subfigure}[b]{0.49\linewidth}
		\includegraphics[width=\linewidth]{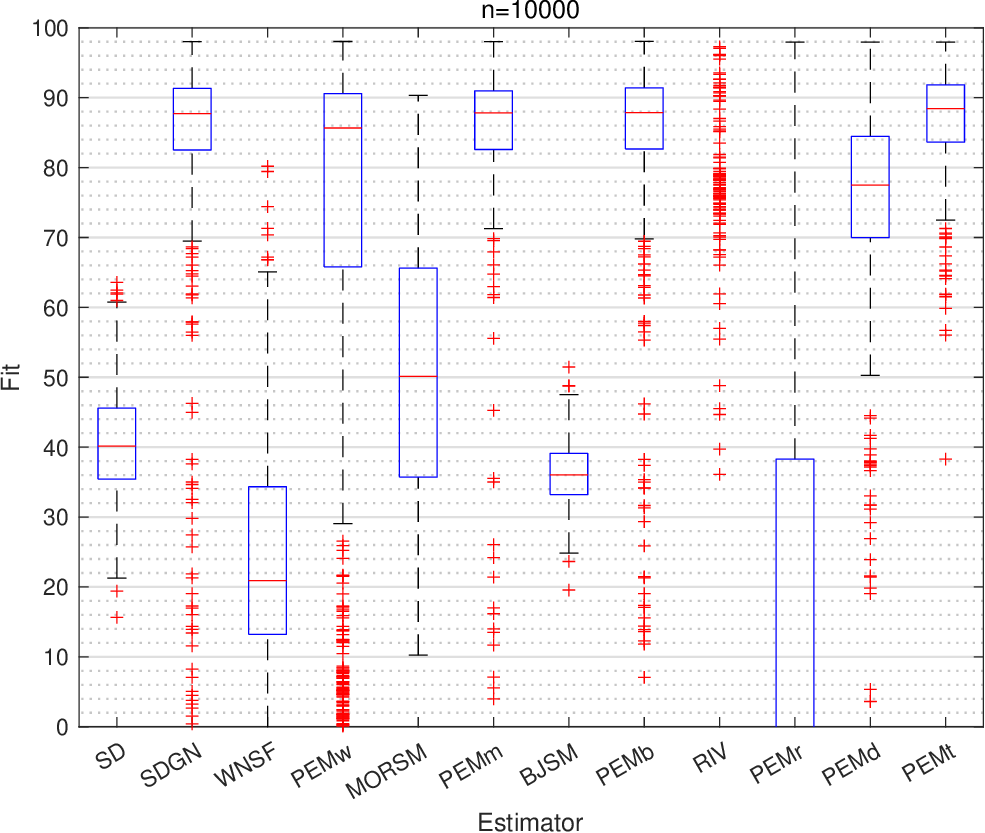}
		\caption{$n=10000$}
		\label{fit_10000_aic}
	\end{subfigure}\hfill
	\begin{subfigure}[b]{0.49\linewidth}
		\includegraphics[width=\linewidth]{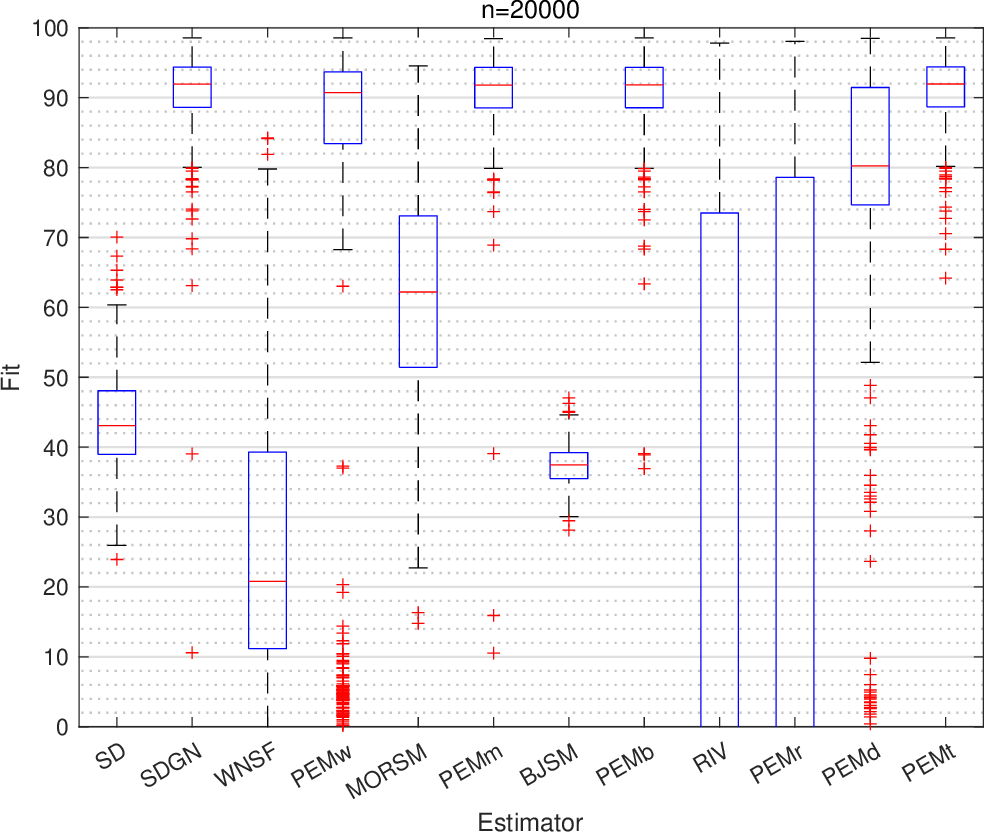}
		\caption{$n=20000$}
		\label{fit_20000_aic}
	\end{subfigure}
	\caption{The boxplot of the fits given by the estimators under different sample sizes.}
	\label{fit_aic}
\end{figure}

{\color{blue}
\subsubsection{Truncation model order selection}
The truncation order \( m \) of the ARX(\(\infty\)) model plays a critical role in the estimation accuracy of the SD, WNSF, MORSM, and BJSM estimators, and thus must be chosen carefully.
From a theoretical standpoint, \( m \) should grow to infinity at a logarithmic rate as the sample size \( n \to \infty \) to balance estimation variance and truncation bias.
For practical implementation, we select \( m \) using the Akaike Information Criterion (AIC) applied to the ARX(\(\infty\)) model \eqref{arx2}. Specifically, we define the AIC-based estimate as
\begin{align}
	\widehat{m} = \argmin_{m=1,2,\cdots,n}~
	n\log\left(\frac1n
	\big\|Y-X\widehat{\theta}_n^{vw}\big\|^2
	\right) + 4m,\!\!\label{m1}
\end{align}
where $\widehat{\theta}_n^{vw}$ is the LS estimate given by \eqref{ls}.
The factor \( 4m \) accounts for the total number of estimated parameters (two polynomials of order \( m \) each).
The value \( \widehat{m} \) obtained from \eqref{m1} is then used uniformly across all four estimators—SD, WNSF, MORSM, and BJSM—to ensure a fair comparison. To reduce computational burden in the simulations, the search for \( \widehat{m} \) is restricted to the grid \( \{10, 20, \dots, 150\} \).

}

\subsubsection{Performance measures}
We evaluate the estimators based on three criteria: estimation accuracy, computational complexity, {\color{blue}and number of iterations required for GN refinement:}
\begin{enumerate}[(i)]
	\item Estimation accuracy: Accuracy is measured using the Fit metric \cite{Ljung2012}, defined as
	      \begin{align*} \mbox{\rm Fit} =100\times \left( 1 -
		      \frac{\|\widehat{\theta}_n -
			      \theta^o \|}{\|\theta^o -\overline{\theta^o}\|}\right),
	      \end{align*}
	      where $\widehat{\theta}_n$ denotes the {\color{blue}estimate} produced by a given estimator,  $\theta^o$ is the true parameter vector of the model \eqref{bjnum}, and $\overline{\theta^o}$ is its arithmetic mean of $\theta^o$;

	\item Computational complexity without GN iterations: We report the running time (in seconds) of the SD, WNSF, MORSM, BJSM, and RIV estimators. {\color{blue}For fairness, the reported times for SD, WNSF, MORSM, and BJSM exclude the shared preprocessing step of estimating the truncation order
	      via \eqref{m1}, as this computation is common to all four methods.
	      In contrast, the running time for the RIV estimator corresponds to the execution of MATLAB’s {\ttfamily rivbj} command and is included for completeness rather than direct comparison, as RIV follows a fundamentally different estimation paradigm.}

	\item {\color{blue}Number of GN iterations: We record the number of GN iterations required for convergence in the refinement stage of the following PEM-based estimators: SDGN, PEMw, PEMm, PEMb, PEMr, PEMd, and PEMt.}

\end{enumerate}


\begin{figure}[ht]
	\centering
	\begin{subfigure}[b]{0.49\linewidth}
		\includegraphics[width=\linewidth]{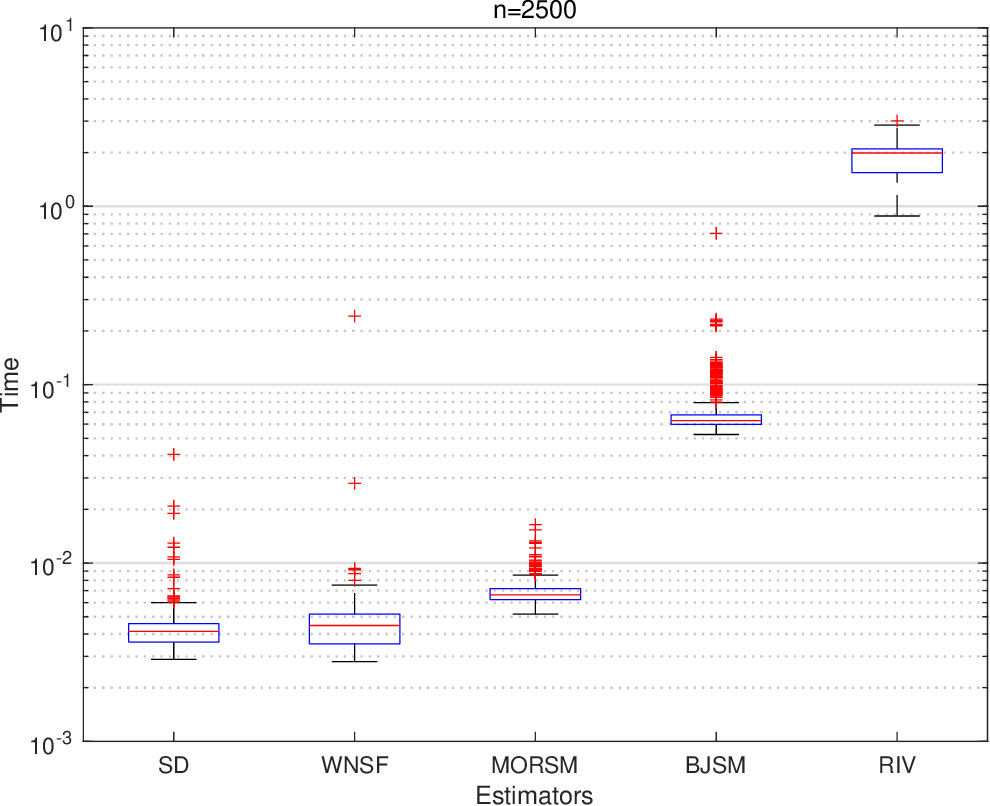}
		\caption{$n=2500$}
		\label{time_2500_aic}
	\end{subfigure}\hfill
	\begin{subfigure}[b]{0.49\linewidth}
		\includegraphics[width=\linewidth]{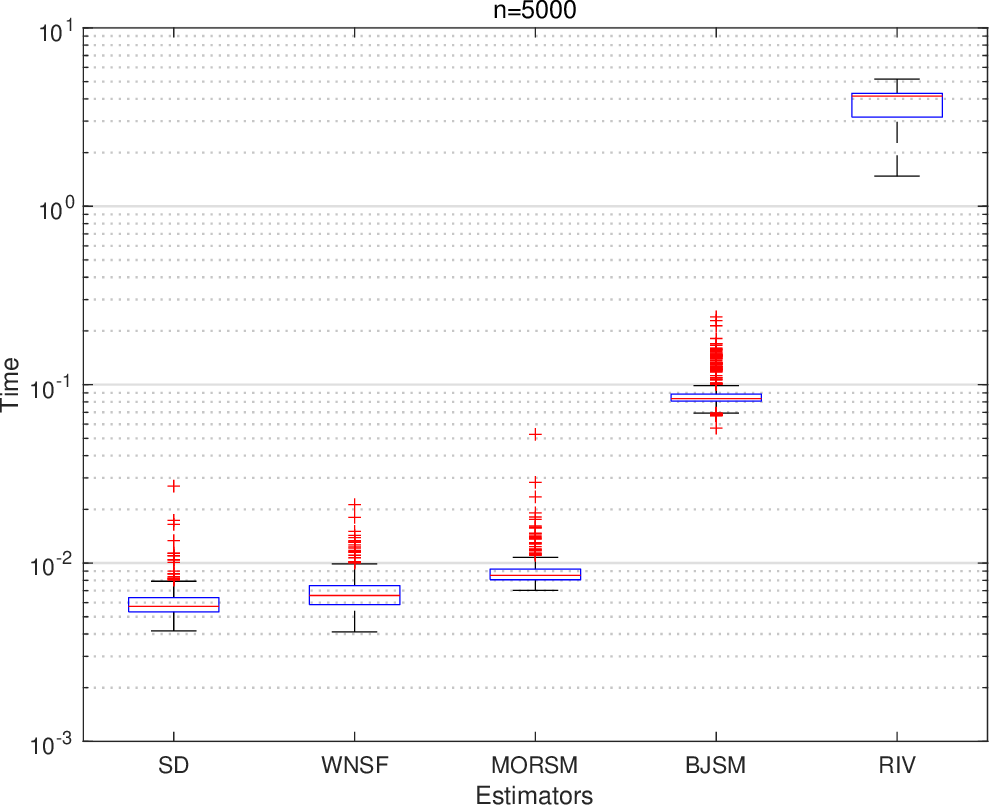}
		\caption{$n=5000$}
		\label{time_5000_aic}
	\end{subfigure}

	\vspace{2mm}

	\begin{subfigure}[b]{0.49\linewidth}
		\includegraphics[width=\linewidth]{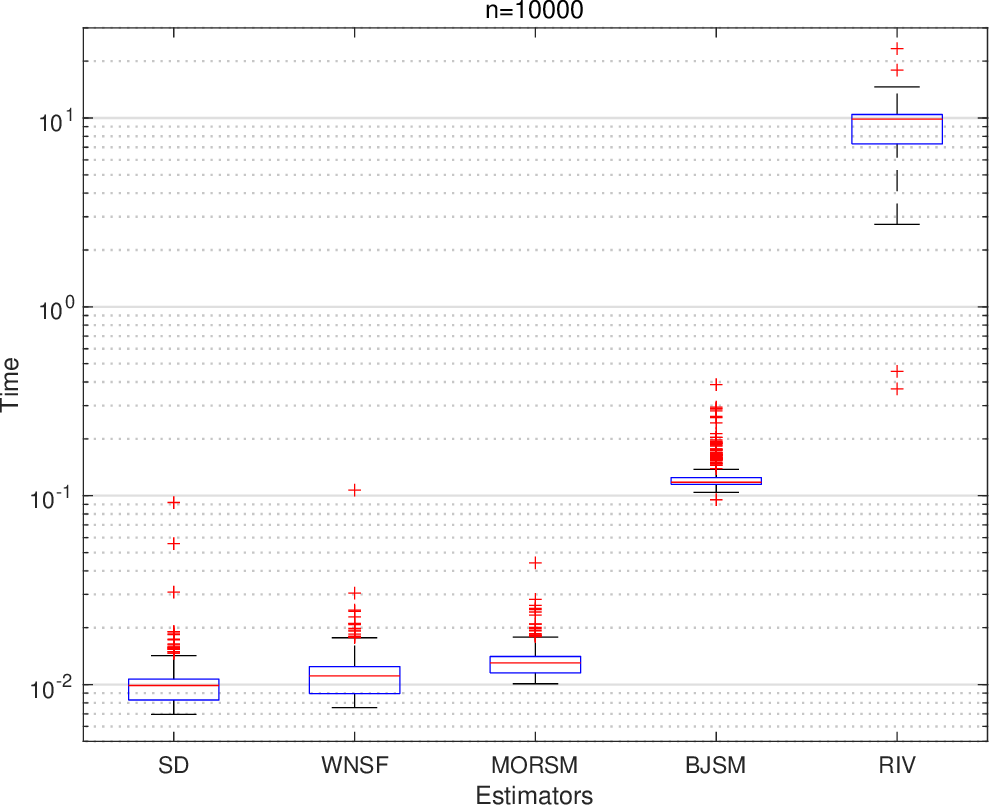}
		\caption{$n=10000$}
		\label{time_10000_aic}
	\end{subfigure}\hfill
	\begin{subfigure}[b]{0.49\linewidth}
		\includegraphics[width=\linewidth]{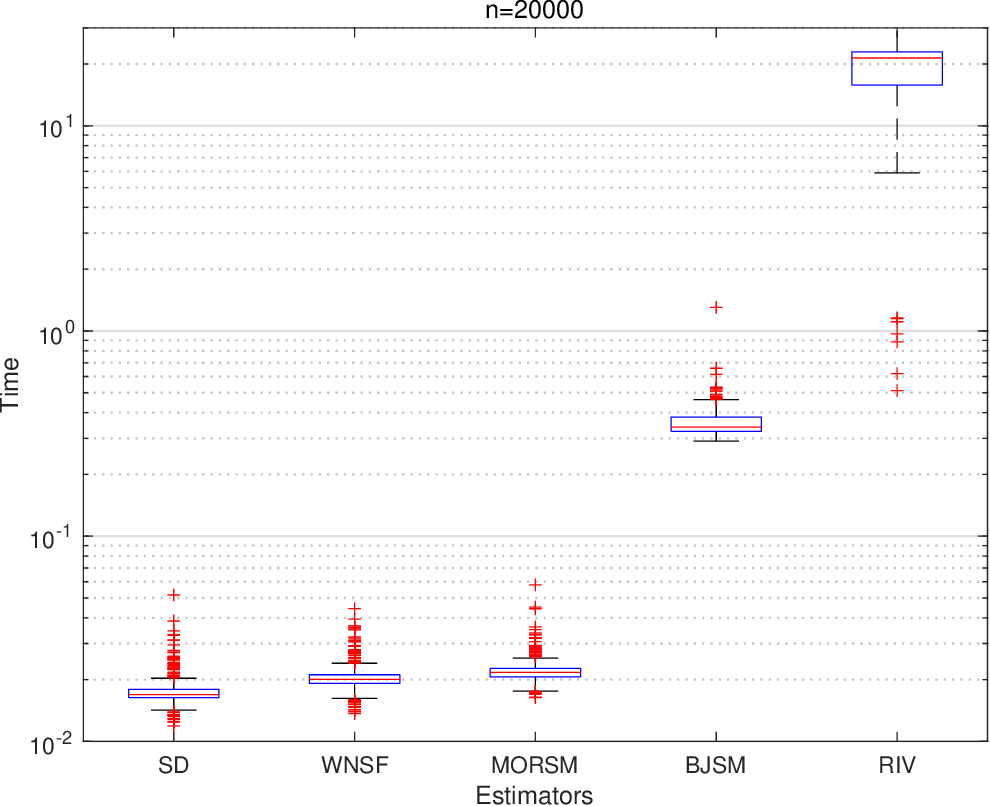}
		\caption{$n=20000$}
		\label{time_20000_aic}
	\end{subfigure}
	\caption{The boxplot of the running times of the estimators without using GN refinement under different sample sizes.}
	\label{time_aic}
\end{figure}
\subsubsection{Simulation results}
{\color{blue}We present simulation results of these estimators according to the estimation accuracy, running time, and number of iterations:
}
\begin{enumerate}[(i)]
	\item The boxplots of the fit values obtained by all estimators are presented in Figs.~\ref{fit_2500_aic}–\ref{fit_20000_aic}, and the corresponding average fits across all simulation cases are summarized in Table~\ref{fit1}.
	      {\color{blue}We observed that the initial estimate of the noise model from the MORSM estimator can be numerically unstable in some realizations. This is because the true noise model has poles at 0.98 and 0.97 (very close to the unit circle), the  one iteration occasionally fails, producing \texttt{NaN} values.
	      To ensure meaningful visualization, these \texttt{NaN} outcomes are excluded from the boxplots and average fits for MORSM. The number of such failures (out of 500 Monte Carlo realizations) is reported in parentheses in the MORSM column of Table~\ref{fit1}.}

	\item
	      The boxplots of the running times for the SD, WNSF, MORSM, BJSM, and RIV estimators are displayed in Figs.~\ref{time_2500_aic}–\ref{time_20000_aic}, and the corresponding average running times are reported in Table~\ref{time}.

	\item
	      {\color{blue}The boxplots of the number of GN iterations for SDGN, PEMw, PEMm, PEMb, PEMr, PEMd, and PEMt are shown in Figs.~\ref{max_iter_2500_aic}–\ref{max_iter_20000_aic}, and the corresponding average iteration counts are reported in Table~\ref{max_iter}.}

\end{enumerate}

\begin{figure}[!ht]
	\centering
	\begin{subfigure}[b]{0.49\linewidth}
		\includegraphics[width=\linewidth]{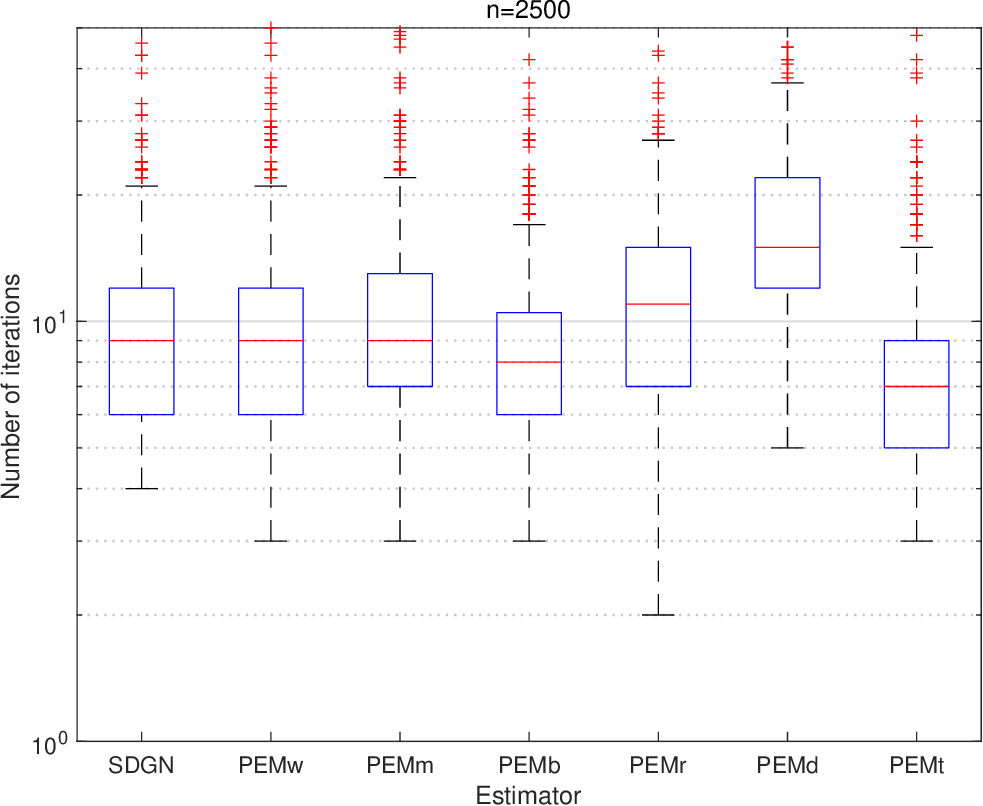}
		\caption{$n=2500$}
		\label{max_iter_2500_aic}
	\end{subfigure}\hfill
	\begin{subfigure}[b]{0.49\linewidth}
		\includegraphics[width=\linewidth]{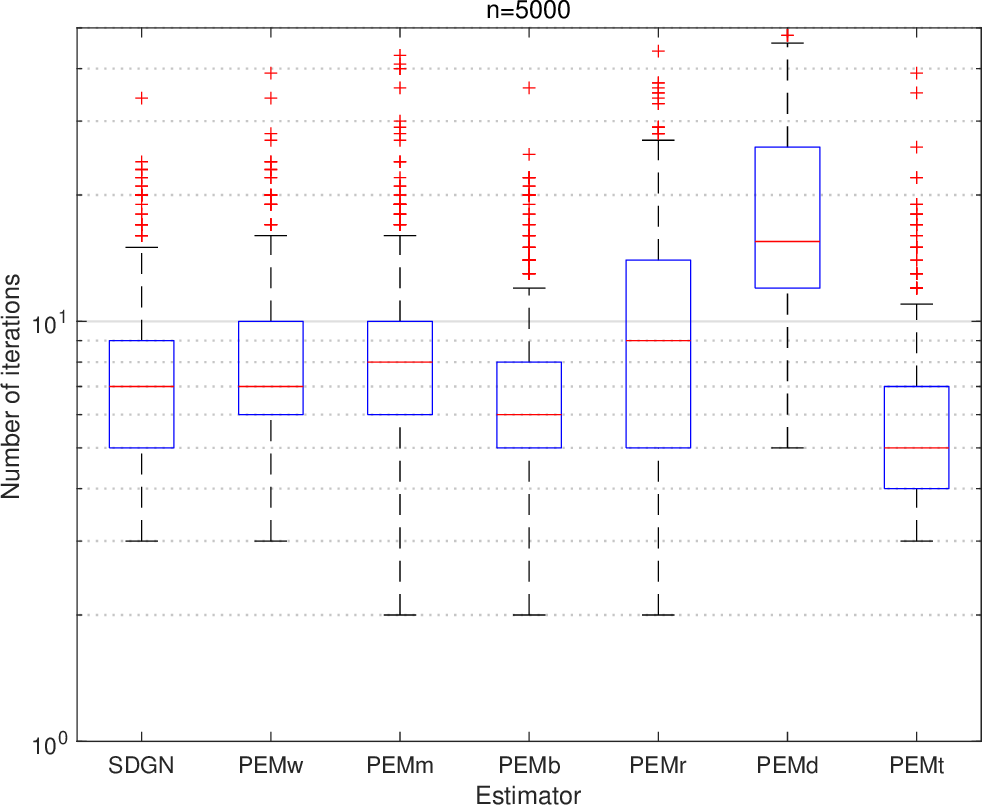}
		\caption{$n=5000$}
		\label{max_iter_5000_aic}
	\end{subfigure}

	\vspace{2mm}

	\begin{subfigure}[b]{0.49\linewidth}
		\includegraphics[width=\linewidth]{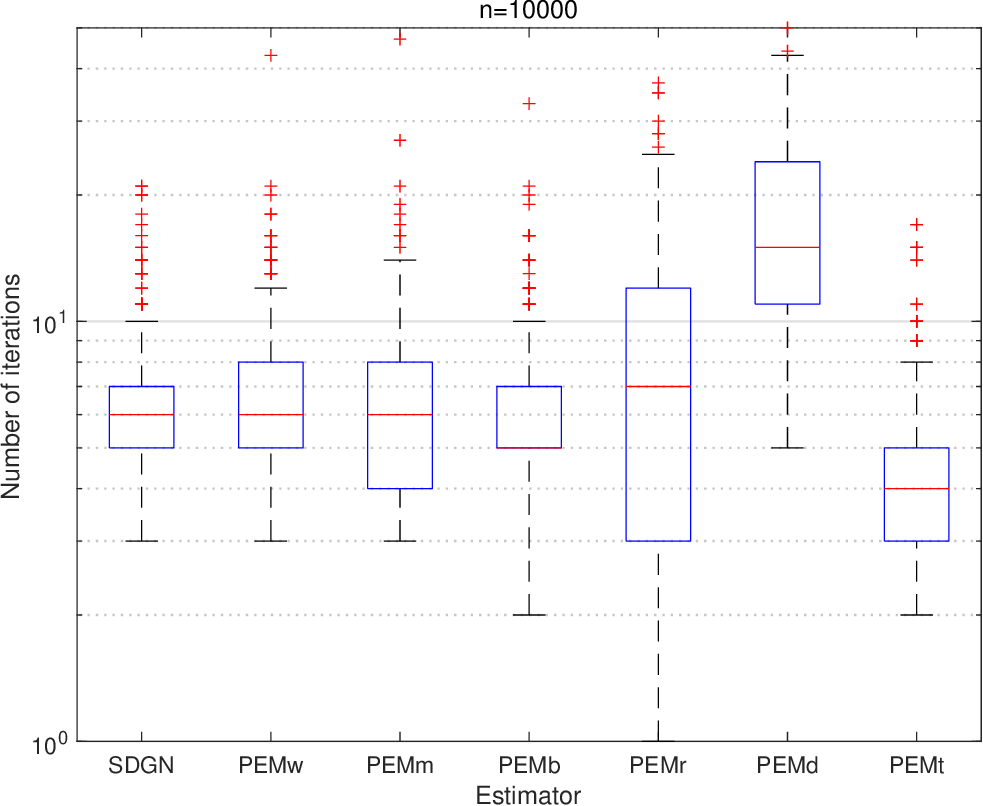}
		\caption{$n=10000$}
		\label{max_iter_10000_aic}
	\end{subfigure}\hfill
	\begin{subfigure}[b]{0.49\linewidth}
		\includegraphics[width=\linewidth]{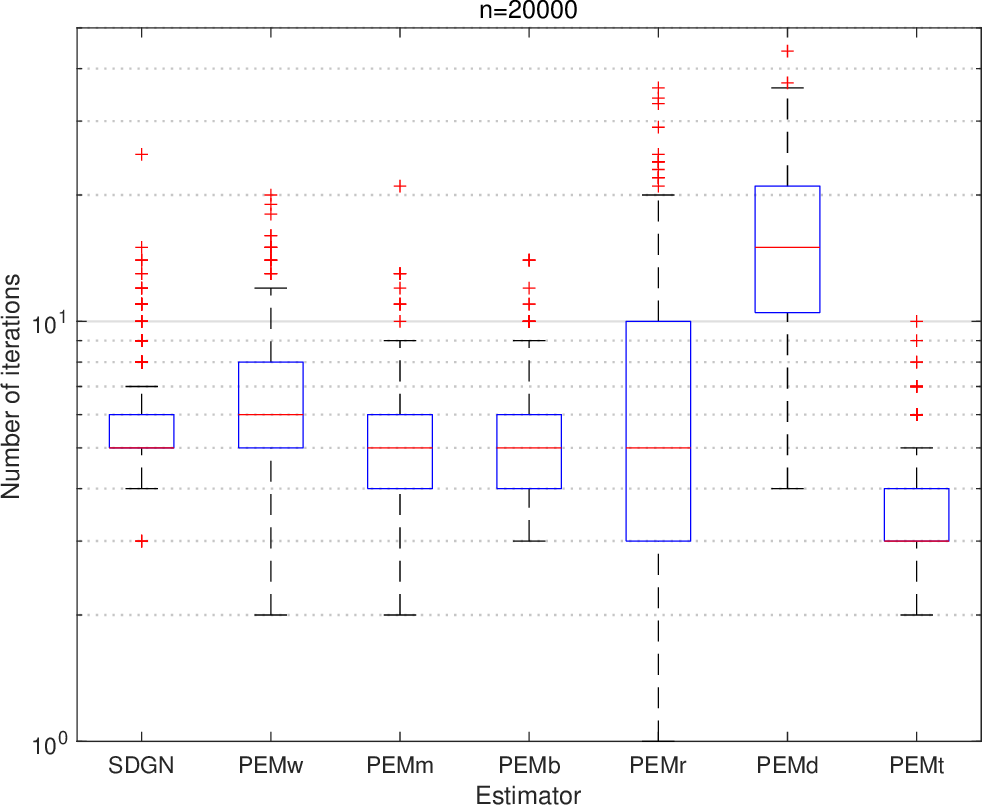}
		\caption{$n=20000$}
		\label{max_iter_20000_aic}
	\end{subfigure}
	\caption{The boxplot of the number of iterations for the estimators using GN refinement under different sample sizes.}
	\label{max_iter_aic}
\end{figure}

\subsubsection{Findings}
Based on the simulation results, we
summarize the   key findings for these estimators:

\begin{enumerate}[(i)]

	\item {\color{red!70!black}The SD estimator consistently outperforms the WNSF, BJSM, and RIV in terms of fit across all sample sizes.}

	\item All five estimators—SD, WNSF, MORSM, BJSM, and RIV—benefit from GN refinement, which consistently improves their estimation accuracy.  {\color{blue}In particular, SDGN, PEMm, and PEMb nearly attain the performance of the best-performing estimator  PEMt as sample size increases, while PEMw and PEMr does not exhibit this behaviour.

	\item The estimation accuracy of all estimators   improves with increasing sample size.

	\item  In terms of computational time, the SD estimator is slightly faster than both WNSF and MORSM across all sample sizes. Overall, SD, WNSF, and MORSM exhibit comparable computational complexity and are more substantially efficient than  BJSM.
	      Moreover, the numbers of GN iterations for the PEM-based estimators: SDGN, PEMw, PEMm, PEMb, and PEMr, show no significant differences across the tested scenarios.
	} 
	
%
\end{enumerate}


{\color{blue}
	\subsection{Random BJ model estimation under low-pass open-loop excitation}
	\label{sec5.3}
	Consider the BJ model  \eqref{gm}
		with the orders $p_b=4,~p_c=2,~p_d=2,~p_f=4$.
		The polynomials are randomly generated in the following way.
		The coefficients of \( B^o(q) \) are drawn independently from a uniform distribution on \([-1, 1]\).
		The roots of \( C^o(q) \), \( D^o(q) \), and \( F^o(q) \) all lie within an annulus: each root’s magnitude is sampled uniformly from \([0.6, 0.95]\), and its phase is sampled uniformly from \([0^\circ, 90^\circ]\), with the corresponding complex conjugate also included to ensure real-valued coefficients.
		
		For each realization, we generate a random BJ model following the way described above and an iid Gaussian white noise
		input sequence with zero mean and unit variance filtered by the transfer function $1/(1-0.8q^{-1})^2$.  And the output \( y(t) \) is simulated by the generated random BJ model  using the filtered input \( u(t) \) and an iid Gaussian white noise sequence \( e(t) \), where the variance of \( e(t) \) is chosen so that  the signal-to-noise  ratio
		\begin{align*}
			\frac{\sum_{t=1}^n \big[B^o(q)/F^o(q)u(t)\big]^2}
			{\sum_{t=1}^n \big[C^o(q)/D^o(q)e(t)\big]^2}
			=5.
		\end{align*}
		In total, we generate 500 independent realizations  with length 20,000.
		
		We evaluate the performance of the estimators described in Section~\ref{est} over 500 Monte Carlo realizations.
		Simulation results are reported for subsample sizes $n=2500 , 5000 , 10000 ,$ and $20000$ to illustrate how the performance of all estimators evolves with increasing sample size.
		For saving space, we
		only report the average (median) fits, average running times, and average number of iterations of these estimators   in Tables \ref{fit3}--\ref{max_iter3}.
		
		The simulation results convey a message largely consistent with Section~\ref{sec5.2}, with one notable difference: all SD, WNSF, MORSM, and BJSM have lower fit than RIV. However, after GN refinement, their performance becomes comparable to that of PEMr for large sample sizes.

	}

\begin{table*}[h]
	\centering
	{\color{blue}
		\caption{The average (median)   fits of all estimators   among 500 realizations  under different sample sizes.}
		\label{fit3}
		\setlength{\tabcolsep}{6pt}
		\resizebox{\textwidth}{!}{%
		\begin{tabular}{cccccccccccccc}
			\hline
			$n$                           & SD                            & SDGN & WNSF & PEMw & MORSM & PEMm & BJSM & PEMb & RIV & PEMr & PEMd & PEMt \\ \hline
			\rule{0pt}{5mm}2500           & \stackanchor{-33.38}{(25.62)} &
			\stackanchor{12.42}{(81.43)}  &
			\stackanchor{-53.20}{(11.75)} &
			\stackanchor{-7.78}{(77.44)}  &
			\stackanchor{-39.98}{(45.33)} &
			\stackanchor{25.79}{(83.11)}  &
			\stackanchor{-24.34}{(37.83)} &
			\stackanchor{10.71}{(82.02)}  &
			\stackanchor{-20.18}{(73.37)} &
			\stackanchor{20.76}{(81.70)}  &
			\stackanchor{4.99}{(74.42)}   &
			\stackanchor{36.19}{(85.32)}                                                                                                               \\[3mm] \hline
			\rule{0pt}{5mm}5000           &
			\stackanchor{-15.05}{(33.23)} &
			\stackanchor{31.71}{(87.00)}  &
			\stackanchor{-39.81}{(19.21)} &
			\stackanchor{26.76}{(85.06)}  &
			\stackanchor{-10.47}{(60.64)} &
			\stackanchor{30.64}{(87.75)}  &
			\stackanchor{-9.07}{(42.42)}  &
			\stackanchor{31.82}{(87.41)}  &
			\stackanchor{23.59}{(85.05)}  &
			\stackanchor{43.60}{(87.43)}  &
			\stackanchor{8.26}{(86.53)}   &
			\stackanchor{64.00}{(88.63)}                                                                                                               \\[3mm] \hline
			\rule{0pt}{5mm}10000          &
			\stackanchor{6.04}{(41.91)}   &
			\stackanchor{43.33}{(91.65)}  &
			\stackanchor{-16.96}{(29.37)} &
			\stackanchor{51.05}{(90.89)}  &
			\stackanchor{17.96}{(76.72)}  &
			\stackanchor{48.33}{(91.86)}  &
			\stackanchor{1.42}{(47.10)}   &
			\stackanchor{45.91}{(91.69)}  &
			\stackanchor{27.97}{(90.70)}  &
			\stackanchor{47.26}{(91.69)}  &
			\stackanchor{47.75}{(90.86)}  &
			\stackanchor{62.52}{(91.95)}                                                                                                               \\[3mm] \hline
			\rule{0pt}{5mm}20000          &
			\stackanchor{24.73}{(50.10)}  &
			\stackanchor{64.29}{(94.43)}  &
			\stackanchor{6.06}{(46.53)}   &
			\stackanchor{60.92}{(94.31)}  &
			\stackanchor{39.52}{(87.28)}  &
			\stackanchor{62.92}{(94.53)}  &
			\stackanchor{7.94}{(47.15)}   &
			\stackanchor{61.01}{(94.53)}  &
			\stackanchor{57.53}{(94.02)}  &
			\stackanchor{62.62}{(94.54)}  &
			\stackanchor{51.65}{(94.09)}  &
			\stackanchor{74.86}{(94.72)}                                                                                                               \\[3mm] \hline
		\end{tabular}}
	}

\end{table*}

\begin{table}[H]
	\centering
	{\color{blue}
		\caption{The average running times of the estimators  without using GN refinement  among 500 realizations under different sample sizes  (Unit: Seconds).}
		\label{time3}
		\begin{tabular}{cccccr}
			\hline
			$n$    & SD     & WNSF & MORSM & BJSM & RIV \\ \hline
			2500   & 0.0045 &
			0.0045 &
			0.0082 &
			0.1138 &
			1.9407                                      \\ \hline
			5000   &
			0.0058 &
			0.0065 &
			0.0100 &
			0.1357 &
			3.2580                                      \\ \hline
			10000  &
			0.0090 &
			0.0086 &
			0.0138 &
			0.1736 &
			5.8904                                      \\ \hline
			20000  &
			0.0138 &
			0.0135 &
			0.0205 &
			0.2386 &
			12.4884                                     \\ \hline
		\end{tabular}
	}
\end{table}

\begin{table}[H]
	\centering
	{\color{blue}
		\caption{The average number of iterations of the estimators using GN refinement  among 500 realizations  under different sample sizes.}
		\label{max_iter3}
		\vspace{1.5ex}
		\begin{tabu} to 0.7
			\textwidth{cccccccc}
			\hline
			$n$   & SDGN & PEMw & PEMm & PEMb & PEMr & PEMd & PEMt \\ \hline
			2500  &
			10.26 &
			10.08 &
			7.83  &
			7.03  &
			5.70  &
			13.56 &
			6.94
			\\ \hline
			5000  &
			8.48  &
			9.33  &
			6.02  &
			5.84  &
			4.24  &
			12.10 &
			5.19                                                   \\ \hline
			10000 &
			7.38  &
			7.58  &
			5.00  &
			5.34  &
			3.33  &
			10.91 &
			4.66                                                   \\ \hline
			20000 &
			6.60  &
			6.52  &
			4.35  &
			4.89  &
			2.59  &
			9.90  &
			3.62                                                   \\ \hline
		\end{tabu}}
\end{table}


\section{Conclusions}
\label{sec6}

{\color{red!70!black}In this paper, we propose a new consistent and asymptotically efficient estimation method, called SDGN, as a possible alternative to the  WNSF approach for BJ model estimation under both open-loop and closed-loop data, particularly in scenarios where the high-order ARX regressor matrix becomes ill-conditioned under low-pass input excitation.}{\color{blue} 
		The SDGN method combines  SD estimator with GN refinement step. The use of GN iterations significantly simplifies the design and theoretical analysis of the initial estimator, as it only requires consistency (not asymptotic efficiency) from the SD stage.
		The SD estimator is rooted in the model reduction framework: it leverages a nonparametric high-order ARX estimate to construct filtered input–output signals and then sequentially recovers the dynamic and noise components of the BJ model by solving two  OE subproblems via LS.
		We establish that the SD estimator is consistent under both open-loop and closed-loop conditions, relying on the consistency of the LS estimator for OE models. Furthermore, one-step GN iteration suffices to refine this initial estimate to asymptotic efficiency, matching the optimal performance of both the WNSF and PEM.}

	{\color{blue}
		Simulation results confirm that the proposed SDGN  outperforms WNSF under low-pass input excitation.
		At the same time, BJSM, MORSM, and RIV also demonstrate strong empirical performance, indicating that their theoretical properties, particularly noise model estimation for BJSM and MORSM and their potential extension to closed-loop settings, deserve further investigation. }

\medskip

\appendix

\textbf{Appendix A: Proofs of main results}
\setcounter{thm}{0}
\renewcommand{\thethm}{A\arabic{thm}}
\setcounter{lem}{0}
\renewcommand{\thelem}{A\arabic{lem}}
\setcounter{rem}{0}
\renewcommand{\therem}{A\arabic{rem}}
\renewcommand{\thesection}{A}

This section includes the proof  of main results in the paper.
\subsection{Proof of Theorem \ref{thm1}}
We prove the theorem by following the steps of the algorithm for the SD estimator.

\subsubsection{Rate of convergence  $ \widehat{\theta}_n^v$ and $\widehat{\theta}_n^{w}$}

Consider the ARX\((\infty)\) model \eqref{sea1}.
We establish the convergence rates of the estimators \(\widehat{\theta}_n^v\) and \(\widehat{\theta}_n^{w}\) by verifying the conditions of Lemma~\ref{lm2}.

{\color{blue}First, the stability of \(D^o(q)\) and \(F^o(q)\), imposed in Assumption~\ref{assum1} for the BJ model \eqref{gm}, implies that the infinite impulse response sequences \(\{v_k^o\}\) and \(\{w_k^o\}\) associated with the ARX\((\infty)\) model \eqref{sea1} decay exponentially. Consequently,
$
	\sum_{k=1}^\infty \sqrt{k}\,|v_k^o| < \infty \quad \text{and} \quad \sum_{k=1}^\infty \sqrt{k}\,|w_k^o| < \infty.
$
Moreover, Assumptions~\ref{assum3}–\ref{assum4} concerning the noise and input for the BJ model \eqref{gm} also hold for the ARX\((\infty)\) model \eqref{sea1}.
Together with the assumption on the truncation order \(m\), these conditions satisfy the requirements of Lemma~\ref{lm2}, thereby yielding the desired convergence rates for \(\widehat{\theta}_n^v\) and \(\widehat{\theta}_n^{w}\):}
\begin{align}
	\big\|\widehat{\theta}_n^{v}- \theta_{v}^o\big\|_1
	=O_p(\delta_n),~\big\|\widehat{\theta}_n^{w}- \theta_{w}^o\big\|_1
	=O_p(\delta_n).\label{rtvw}
\end{align}

\subsubsection{Rate of convergence of $ \widehat{\theta}_n^b$ and $ \widehat{\theta}_n^{f}$}
\label{A12}

Consider the OE model \eqref{sea2} {\color{blue} with input \(u_V^f(t)\), noise-free output \(y_V^{of}(t)\), and observed output \(y_V^f(t)\).
In the following, we verify that the assumptions required for Lemma~\ref{lm3} hold for this OE model by successively checking the corresponding conditions originally imposed on the BJ model.

First, parts (ii) and (iii) of Assumption~\ref{assum1} imply condition (i) of Lemma~\ref{lm3}.
Second, Assumption~\ref{assum3} ensures condition (ii) of Lemma~\ref{lm3}.
Third, by (ii) and (iv) of Assumption~\ref{assum1} as well as Assumption~\ref{assum4}, the input \(u_V^f(t)\) is persistently exciting of order \(p_f + p_b\), and the regressor
$
	\phi(t) = \big[ -y_V^{of}(t-1), \dots, -y_V^{of}(t-p_f),\, u_V^f(t-1), \dots, u_V^f(t-p_b) \big]^\top
$
associated with the OE model \eqref{sea2} is uncorrelated with the noise \(e(t)\), thereby verifying condition (iii) of Lemma~\ref{lm3}.

Finally, we establish condition (iv) of Lemma~\ref{lm3} for the OE model \eqref{sea2}.
Specifically, we show that the estimated signals
$
	\widehat{u}_V^f(t) \triangleq V(q,\widehat{\theta}_n^{v})u(t), ~
	\widehat{y}_V^{of}(t) \triangleq W(q,\widehat{\theta}_n^{w})u(t), ~
	\widehat{y}_V^f(t) \triangleq V(q,\widehat{\theta}_n^{v})y(t)
$
converge in probability to their true counterparts at the rate \(O_p(\delta_n)\).}

For the estimated sequence $\{\widehat{u}_V^f(t)\}$, by the rate of convergence rate \eqref{rtvw}, we have for all $t=1,\cdots,n$
\begingroup
\allowdisplaybreaks
\begin{align}
	\nonumber
	  |\widehat{u}_V^f(t)-u_V^f(t)|
	&=|\big(V(q,\widehat{\theta}_n^{v}) - V^o(q)\big)u(t)|          \\
	\nonumber
	 & =\Big|\sum_{k=1}^m
	\big(\widehat{v}_k - v^o_k\big) u(t-k) - \sum_{k=m+1}^\infty
	v^o_k u(t-k)\Big|                                              \\
	\nonumber
	\nonumber
	 & \leq \sum_{k=1}^m
	\big|\widehat{v}_k - v^o_k\big| |u(t-k)|+\sum_{k=m+1}^\infty
	\big|v^o_k\big| | u(t-k)|                                      \\
	\nonumber
	 & \leq \sum_{k=1}^m
	\big|\widehat{v}_k - v^o_k\big|O_p(1) +\sum_{k=m+1}^\infty
	\big|v^o_k\big| O_p(1)                                         \\
	\nonumber
	 & =\|\widehat{\theta}_n^{v}- \theta_{v}^o\|_1O_p(1) +O_p(d_m) \\
	 & = O_p(\delta_n) +O_p(d_m) =O_p(\delta_n).\label{a1}
\end{align}
\endgroup
For the estimated sequence $\{\widehat{y}_V^{of}(t)\}$,  by the rate of convergence rate \eqref{rtvw} again, similarly there holds that for all $t=1,\cdots,n$
\begingroup
\allowdisplaybreaks
\begin{align}
	|\widehat{y}_V^{of}(t)  -y_V^{of}(t)|  =|\big(W(q,\widehat{\theta}_n^{v}) - W^o(q)\big)u(t)| =O_p(\delta_n).\label{a12}
\end{align}
\endgroup
For the estimated sequence $\{\widehat{y}_V^f(t)\}$,  by the rate of convergence rate \eqref{rtvw} again,   there holds that for all $t=1,\cdots,n$
\begingroup
\allowdisplaybreaks
\begin{align}
	|\widehat{y}_V^f(t)  -y_V^f =|\big(V(q,\widehat{\theta}_n^{v}) - V^o(q)\big)y(t)| =O_p(\delta_n)\label{a13}
\end{align}
\endgroup
by noting   $y(t)=O_p(1)$.
Therefore,  Lemma \ref{lm3} yields
\begingroup
\allowdisplaybreaks
\begin{subequations}
	\label{rt}
	\begin{align}
		\label{rtb}
		 & \|\widehat{\theta}_n^{b}-\theta_{b}^o\|_2=\max\{O_p(\delta_n),O_p(1/\sqrt{n})\}
		=O_p(\delta_n),                                                                    \\
		 & \|\widehat{\theta}_n^{f}-\theta_{f}^o\|_2
		=\max\{O_p(\delta_n),O_p(1/\sqrt{n})\}
		=O_p(\delta_n).
		\label{rtf}
	\end{align}
\end{subequations}
\endgroup

\subsubsection{Rate of convergence of $\widehat{\theta}^c_n$
	and $\widehat{\theta}^{d}_n$}

Consider the OE model \eqref{sea3} {\color{blue} with input \(u_{BF}^f(t)\), noise-free output \(y_V^{of}(t)\), and observed output \(y_V^f(t)\).
In the following, we verify that the conditions of Lemma~\ref{lm3} hold for this OE model by leveraging the assumptions originally imposed on the BJ model.

Similar to Section~\ref{A12}, conditions (i)–(iii) of Lemma~\ref{lm3} can be verified for the OE model \eqref{sea3} using Assumptions~\ref{assum1}–\ref{assum4}.
Consequently, it remains only to verify condition (iv) of Lemma~\ref{lm3} for the estimated signals
$
	\widehat{u}_{BF}^f(t),$ ~ $\widehat{y}_V^{of}(t),$    and  $\widehat{y}_V^f(t)
$
associated with the OE model \eqref{sea3}.

Since the convergence rates of \(\widehat{y}_V^{of}(t)\) and \(\widehat{y}_V^f(t)\) have already been established in Section~\ref{A12}, it suffices to establish the convergence rate of the sequence \(\widehat{u}_{BF}^f(t)\).
This is proved by}
\begingroup
\allowdisplaybreaks
\begin{align*}
	\nonumber
	 |\widehat{u}_{BF}^f(t)-u_{BF}^f(t)| &=\left|\left(\frac {B(q,\widehat{\theta}^b_n)}{F(q,\widehat{\theta}^f_n)} - \frac{B^o(q)}{F^o(q)}\right)u(t)\right| \\
	\nonumber
	 & \leq\left|\left(\frac {B(q,\widehat{\theta}^b_n)}{F(q,\widehat{\theta}^f_n)} - \frac{B(q,\widehat{\theta}^b_n)}{F^o(q)}\right)u(t)\right|             +\left|\left(\frac {B(q,\widehat{\theta}^b_n)}{F^o(q)} - \frac{B^o(q)}{F^o(q)}\right)u(t)\right|                                              \\
	\nonumber
	 & =\left|\left(\frac {1}{F(q,\widehat{\theta}^f_n)} - \frac{1}{F^o(q)}\right)B(q,\widehat{\theta}^b_n)u(t)\right|                                        +\left|\big( B(q,\widehat{\theta}^b_n)  -  B^o(q)\big)\frac1{F^o(q)}u(t)\right|                                                               \\
	 & =O_p(\delta_n)
\end{align*}
\endgroup
for all $t=1,\cdots,n$.
Here, two upper bounds in probability are used.  One is
\begingroup
\allowdisplaybreaks
\begin{align*}
	   \left|\left(\frac {1}{F(q,\widehat{\theta}^f_n)} - \frac{1}{F^o(q)}\right)B(q,\widehat{\theta}^b_n)u(t)\right|  
	 & =
	\left|\big(F(q,\widehat{\theta}^f_n) - F^o(q)\big)
	\frac{B(q,\widehat{\theta}^b_n)}{F(q,\widehat{\theta}^f_n)F^o(q)}u(t)\right|                                      \\
	 & \leq\|\widehat{\theta}_n^{f}-\theta_{f}^o\|_2O_p(1)=O_p(\delta_n).
\end{align*}
\endgroup
Another is
\begingroup
\allowdisplaybreaks
\begin{align*}
	   \left|\big( B(q,\widehat{\theta}^b_n)  -  B^o(q)\big)\frac1{F^o(q)}u(t)\right|  
	 & \leq \sqrt{\sum_{k=1}^{p_b}
	\big(\widehat{b}_k - b^o_k\big)^2} \sqrt{\sum_{k=1}^{p_b}u_F^f(t-k)^2}            \\
	 & =\|\widehat{\theta}_n^{b}-\theta_{b}^o\|_2O(1)=O_p(\delta_n)
\end{align*}
\endgroup
by the Cauchy-Swarchz inequality, where $u_F^f(t) \eq \frac1{F^o(q)}u(t)$.

Therefore,  Lemma \ref{lm3} yields
\begin{align}
	\label{rtd}
	\|\widehat{\theta}_n^{d}-\theta_{d}^o\|_2
	=O_p(\delta_n), ~
	\|\widehat{\theta}_n^{c}-\theta_{c}^o\|_2
	=O_p(\delta_n).
\end{align}

\subsubsection{Rate of convergence of the SD estimator  }

Combining the rates \eqref{rt} and 	\eqref{rtd} achieves
\begin{align}
	\|\widehat{\theta}_n^{\rm sd} - \theta^o\|_2 = O_p(\delta_n).
\end{align}
This completes the proof.

\subsection{Proof of Proposition \ref{thm2}}

The proof is straightforward using the following formulas:
\begin{align*}
	\rho^{\alpha \log (n)} = n^{\alpha \log (\rho)}, \quad
	\rho^{\alpha \log \log (n)} = (\log n)^{\alpha \log (\rho)}.
\end{align*}
Accordingly, the details are omitted.

\section*{Appendix B: Proofs of auxiliary lemmas}

\renewcommand{\thesection}{B}

\setcounter{equation}{0}
\setcounter{lem}{0}
\renewcommand{\thelem}{B\arabic{lem}}
\setcounter{rem}{0}
\renewcommand{\therem}{B\arabic{rem}}

This appendix contains the proofs of Lemmas \ref{lm2} and \ref{lm3} that are applied to proving Theorem \ref{thm1}.
\setcounter{subsection}{0}
\subsection{Proof of Lemma \ref{lm2}}

{\color{red!70!black}
First, by \cite[Theorem 6.1]{Ljung1992}, under Assumptions \ref{assum3} and \ref{assum4} and the conditions on the truncation order $m$, we have
\begin{align}
	E\big(\big\|\widehat{\theta}_n^{vw}-\bar{\theta}_{vw}\big\|_1\big)
	=O
	\left(
	\frac{m}{\sqrt{n}} +d_mm
	\sqrt{\frac{\log(n)}{n}}
	\right),\label{a11}
\end{align}
where $\bar{\theta}_{vw}$ is defined in Eq.~(5.1) of \cite{Ljung1992}.
Second, by \cite[Lemma 5.1]{Ljung1992}, under the same assumptions, we have
}
{\color{blue}
\begin{align}
	\big\|\bar{\theta}_{vw}-\theta^o_{vw}\big\|_1 =O(d_m).\label{a22}
\end{align}
By combining \eqref{a11} with \eqref{a22}, we have
\begin{align*}
	E\big(  \big\|\widehat{\theta}_n^{vw}-\theta^o_{vw}\big\|_1\big)
&	\leq
	E\big(\big\|\widehat{\theta}_n^{vw}-\bar{\theta}_{vw}\big\|_1\big)
	+
	\big\|\bar{\theta}_{vw}-\theta^o_{vw}\big\|_1                     \\
	       & =O
	\left(
	\frac{m}{\sqrt{n}} +d_m\frac{m}{\sqrt{n}} {\sqrt{\log(n)}}
	\right) +O(d_m)                                                   \\
	       & =O
	\left(
	\frac{m}{\sqrt{n}} +d_m
	\right)=O(\delta_n)
\end{align*}
due to $m^{3+\kappa}/n\xra{}0$.
Thus, there exists a constant $C$ such that  $E\big(\big\|\widehat{\theta}_n^{vw}-\theta^o_{vw}\big\|_1/\delta_n\big)
	\leq C $.
By Markov inequality, for any $\epsilon>0$ choose $M>C/\epsilon$.
It takes place
\begin{align*}
	P
	\left(
	\frac{\|\widehat{\theta}_n^{vw}-\theta^o_{vw}\big\|_1}{\delta_n}>M
	\right)\leq \frac{C}{M}<\epsilon.
\end{align*}
This means that $ \|\widehat{\theta}_n^{vw}-\theta^o_{vw}\big\|_1=O_p(\delta_n)$.
}

\subsection{Proof of Lemma \ref{lm3}}

We rewrite the linear model \eqref{lmoe} as
\begin{align*}
	\widehat{y}(t) = \widehat{\phi}(t)^\top \theta_{fb}^o +
	\underbrace{e(t) + \overbrace{\widehat{y}(t) - y(t) + \big(\phi(t) - \widehat{\phi}(t)\big)^\top \theta_{fb}^o}^{\chi(t)}}_{\varpi(t)},
\end{align*}
where $\chi(t) = O_p(\zeta_n)$ by condition (iv) of Lemma~\ref{lm3}.
Applying condition (iv) once again, we obtain
\begingroup
\allowdisplaybreaks
\begin{align}
	\frac1n\widehat{\Phi}^T\varpi
	 & =\frac1n\sum_{t=1}^n
	\begin{pmatrix}
		-\widehat{y}^o(t-1)   \\
		\vdots                \\
		-\widehat{y}^o(t-p_f) \\
		\widehat{u}(t-1)      \\
		\vdots                \\
		\widehat{u}(t-p_b)
	\end{pmatrix}
	\big(e(t)+\chi(t)\big)  =\frac1n\sum_{t=1}^n
	\begin{pmatrix}
		-y^o(t-1)e(t)   \\
		\vdots          \\
		-y^o(t-p_f)e(t) \\
		u(t-1)e(t)      \\
		\vdots          \\
		u(t-p_b)e(t)
	\end{pmatrix}
	+O_p(\zeta_n).\label{c3}
\end{align}
\endgroup

{\color{blue}By conditions (ii) and (iii) of Lemma~\ref{lm3}, the regressor $\{\phi(t)\}$ is uncorrelated with the noise sequence $\{e(t)\}$. Consequently, for any fixed $1 \leq i \leq p_f$ and $1 \leq j \leq p_b$, the sequences $\{-y^o(t-i)e(t)\}$ and $\{u(t-j)e(t)\}$ are zero-mean, uncorrelated, and have finite variance.  }
Therefore, by Lemma~\ref{lmc3}, each component of the first term in \eqref{c3} is $O_p(1/\sqrt{n})$.
	{\color{blue}It follows that}
\begin{align*}
	\frac1n\widehat{\Phi}^T\varpi = \max\{O_p(\zeta_n),O_p(1/\sqrt{n})\}.
\end{align*}
{\color{blue}By \cite[Lemma 13.1]{Ljung1999}, the persistent excitation of the input $u(t)$ of order $p_f + p_b$ implies that there exists an integer $n_0 > 0$ such that for all $n > n_0$,}
\begin{equation*}
	\frac{1}{n} \sum_{t=1}^n \phi(t)\phi(t)^\top > 0.
\end{equation*}
{\color{blue}Moreover, by condition (iv) of Lemma~\ref{lm3}, we have
$
	\frac{1}{n} \widehat{\Phi}^\top \widehat{\Phi} > 0.
$
It follows that
\begin{align*}
	\widehat{\theta}_n^{fb}
	 & =  ( \widehat{\Phi}^T\widehat{\Phi}
	)^{-1} \widehat{\Phi}^T\widehat{y}
	= ( \widehat{\Phi}^T\widehat{\Phi}
	)^{-1} \widehat{\Phi}^T
	\big(\widehat{\Phi}\theta_{fb}^o+\varpi \big)                  \\
	 & =\theta_{fb}^o+ \Big(\frac1n \widehat{\Phi}^T\widehat{\Phi}
	\Big)^{-1} \Big(\frac 1n\widehat{\Phi}^T \varpi\Big).
\end{align*}}
This means that
$
	\|	\widehat{\theta}_n^{fb} - \theta_{fb}^o\|_2
	=\max\{O_p(\zeta_n),O_p(1/\sqrt{n})\}.
$


\begin{lem}\cite[Theorem 14.4-1 on page 476]{Bishop2007}
	\label{lmc3}
	Let $X(1),X(2), \cdots$ be a sequence of uncorrelated random variables with $E(X(t))=0$ and $E(X(t)^2)<\infty$.
	Thus
	$\frac1n\sum_{t=1}^nX(t)=O_p(1/\sqrt{n}).$

\end{lem}

\end{document}